\documentclass[journal,twoside,web]{ieeecolor}
\usepackage{generic}
\usepackage{cite}
\usepackage{amsmath,amssymb,amsfonts}
\usepackage{algorithmic}
\usepackage{graphicx}
\usepackage{algorithm,algorithmic}
\usepackage{hyperref}
\hypersetup{hidelinks=true}
\usepackage{textcomp}
\usepackage{mathtools}
\usepackage{comment}

\def\BibTeX{{\rm B\kern-.05em{\sc i\kern-.025em b}\kern-.08em
    T\kern-.1667em\lower.7ex\hbox{E}\kern-.125emX}}
\usepackage{algorithm}
\usepackage{algorithmic}

\makeatletter
\newcommand\fs@norules{\def\@fs@cfont{\bfseries}\let\@fs@capt\floatc@ruled
  \def\@fs@pre{}%
  \def\@fs@post{}%
  \def\@fs@mid{\kern3pt}%
  \let\@fs@iftopcapt\iftrue}
\makeatother
\floatstyle{norules}
\restylefloat{algorithm}

\newtheorem{theorem}{Theorem}
\newtheorem{proposition}{Proposition}
\newtheorem{definition}{Definition}
\newtheorem{lemma}{Lemma}
\newtheorem{remark}{Remark}

\newcommand{\ip}[2]{\ensuremath\left\langle{#1}, {#2}\right\rangle}

\newcommand{\dom}{\ensuremath\mathrm{dom}\,}

\newcommand{\norm}[1]{\left \lVert #1 \right \rVert}

\begin{document}

\renewcommand{\arraystretch}{1.4} 
\setlength{\arraycolsep}{0.6em} 

\title{Operator-Splitting Methods for Neuromorphic Circuit Simulation}
\author{Amir Shahhosseini, Thomas Chaffey, and Rodolphe Sepulchre \IEEEmembership{Fellow, IEEE}
\thanks{This work was supported by the European Research Council under
the Advanced ERC Grant Agreement SpikyControl n. 101054323. The
work of Thomas Chaffey was supported by Pembroke College, Cambridge.}
\thanks{Amir Shahhosseini is with the Department
of Electrical Engineering (STADIUS), KU Leuven, B-3001 Leuven, Belgium (e-mail: amir.shahhosseini@kuleuven.be).}
\thanks{T. Chaffey is with the School of Electrical and Computer Engineering, University of Sydney, Australia (e-mail: thomas.chaffey@sydney.edu.au).}
\thanks{Rodolphe Sepulchre is with both the Department of Engineering, University
of Cambridge, CB2 1PZ Cambridge, U.K., and also with the Department
of Electrical Engineering (STADIUS), KU Leuven, B-3001 Leuven, Belgium (e-mail: rodolphe.sepulchre@kuleuven.be).}}

\maketitle
\begin{abstract}
A novel splitting algorithm is proposed for the numerical simulation of neuromorphic circuits. The algorithm is grounded in the operator-theoretic concept of monotonicity, which bears both physical and algorithmic significance. The splitting exploits this correspondence to translate the circuit architecture into the algorithmic architecture. The paper illustrates the many advantages of the proposed operator-theoretic framework over conventional numerical integration for the simulation of multiscale hierarchical events that characterize neuromorphic behaviors. 
\end{abstract}

\begin{IEEEkeywords}
Neuromorphic Systems, Circuit Theory, Convex Optimization, Nonlinear Dynamics.
\end{IEEEkeywords}

\section{Introduction}

The numerical simulation of biophysical neuronal circuits, modelled as interconnections of memristive elements, has a long history and can be traced back to the seminal modelling work of Hodgkin and Huxley \cite{hodgkin1952quantitative}. The interest in this long-standing question is undergoing a strong resurgence in recent years due to the development of neuromorphic engineering applications \cite{gallego2020event,burr2017neuromorphic}, neuromorphic computing \cite{khacef2023spike,young2019review}, spiking neural networks \cite{taherkhani2020review}, and neuromorphic control \cite{ribar2021neuromorphic}. In all those disciplines, numerical simulation is a necessary complement to analysis and design. Nevertheless, simulating neuronal models has always been a challenge: the circuit differential equations are nonlinear and multiscale, resulting in {\textit{stiff}} and {\textit{large-scale}} ordinary differential equations (ODEs). Moreover, the models usually depend on numerous poorly known and highly variable parameters. The mere sensitivity analysis of those models is known to be a delicate question, even though sensitivity analysis is often the very motivation to develop simulation models \cite{drion2011modeling}.

Motivated by the analysis and design of such neuromorphic circuits, this paper explores the potential of operator-theoretic methods for circuit simulation, as an alternative to the standard numerical integration (NI) methods. Operator-theoretic methods and the related splitting algorithms have gained popularity in systems and control in recent decades for a number of reasons: first, they can handle non-smooth and/or multivalued problems, which are hard to solve with numerical integration, see e.g. \cite{brogliato2016nonsmooth}. Second, they are rooted in circuit theory and, in particular, in the concept of monotonicity \cite{minty1960monotone}, making them suitable to acknowledge physical properties such as passivity \cite{camlibel2016linear} and to exploit the circuit topology in the splitting architecture of the algorithms \cite{chaffey2023monotone,chaffey2023circuit}. Finally, they have become the backbone of large-scale optimization algorithms due to the close relationship between monotonicity and convexity \cite{ryu2022large,rockafellar2015convex}.

Those reasons have motivated the recent dissertation \cite{chaffey2022input} and the development of efficient numerical algorithms for the simulation and analysis of electrical circuits composed of monotone elements \cite{chaffey2023monotone,chaffey2023circuit}. As a continuation of \cite{chaffey2022input}, the objective of the present paper is to demonstrate the potential of this framework for neuromorphic circuits. Our angle of attack is to exploit the mixed architecture of neuromorphic systems, which create excitable and spiking behaviors by systematically interlacing elements of positive and negative conductance \cite{sepulchre2022spiking}.  From an operator-theoretic viewpoint, solving the resulting circuit equations amounts to finding the zero of a \emph{difference} of monotone operators. Moving from the solution of a monotone operator to a difference of monotone operators is algorithmically analogous to optimizing a difference of convex functions rather than a convex function, hence a fundamental jump in algorithmic complexity \cite{yuille2003concave}. A key message of the present paper, however, is that this departure is under tight control when solving a neuromorphic system: the number of  ``negative conductance” elements of a neuromorphic circuit is normally much lower than the number of  ``positive conductance” elements \cite{ren2011sodium}, \cite[section II-E]{sepulchre2022spiking}, Furthermore, the splitting between monotone and anti-monotone is physical and interpretable rather than arbitrary. A key appeal of operator-theoretic methods for the simulation of neuromorphic circuits lies indeed in the significance of monotonicity as a both physical and algorithmic property.

The key contributions of this paper are as follows: (i) It is shown how the architecture of a general neuromorphic circuit leads to a systematic decomposition of the operator to difference of monotone suboperators, completely determined by the circuit elements and by the circuit topology; (ii) a splitting algorithm to solve the resulting equations is designed and includes convergence analysis; and (iii) the salient properties of the proposed approach relative to conventional numerical integration are illustrated. We demonstrate the potential of operator-theoretic methods for (i) continuation methods (of key significance for parametric sensitivity analysis and neuromodulation studies), (ii) modular circuit simulation through hierarchical simulation of motifs at increasing scales. This permits the coarse-grained algorithms to accurately detect events of commensurate temporal resolution.

The rest of the paper is organised as follows: Section 2 presents the architecture of a general neuromorphic circuit and derives the corresponding splitting of the operator to be solved. This provides an operator-theoretic representation of spiking networks. Section 3 presents the novel splitting algorithm and several technical details to enable an efficient algorithmic implementation. Section 4 provides three illustrations of the method and discusses its advantages relative to conventional numerical integration. Section 5 concludes with an overview of the paper and directions for future research.

\section{Modeling Neuromorphic Systems}

\subsection{Mathematical Preliminaries}

It is beneficial to begin with a brief review of the core mathematical concepts and definitions that are frequently used within the context of this article. Here, $\mathcal{H}$ denotes a Hilbert space, equipped with an inner product defined over the time axis $\mathbb{T}$, given by
\begin{equation}{\label{Eq:InnerProductDefinition}}
    \ip{u(t)}{y(t)} = \int_{\mathbb{T}} u^\top(t)y(t)dt
\end{equation}
where $u(t), y(t)$ are signals in $\mathcal{H}$.

A central space of interest is $\mathcal{L}^2_\mathbb{T}$, the Hilbert space of square-integrable signals, i.e., signals of finite energy such that $\ip{u(t)}{u(t)} < \infty$. The time domain $\mathbb{T}$ may vary depending on the problem setting.

In order to develop the distributed version of splitting algorithms, it is often necessary to \emph{lift} the problem from a space containing a single signal to a higher-dimensional space containing multiple signals. This lifted Hilbert space is defined using the direct sum:
\begin{equation}
    \boldsymbol{\mathcal{L}}^2 = \bigoplus_{i = 1}^{n} \mathcal{L}^2,
\end{equation}
where  \( \boldsymbol{\mathcal{L}}^2 \) denotes the lifted Hilbert space formed by the direct sum of $n$ copies of $\mathcal{L}^2$, with each copy corresponding to the space containing one of the $n$ signals. In the literature, this space is often referred to as a \emph{product }Hilbert space, particularly in the context of convex optimization and operator splitting methods, where problems are decomposed into finite components \cite{bauschke_convex_2011}. The relationship between lifting in convex optimization and the element extraction method of circuit analysis was studied in \cite{chaffey2023circuit}. In the finite-dimensional case, the direct sum of Hilbert spaces and their Cartesian product become isomorphic as Hilbert spaces.

\begin{definition}\label{def:Monotonicity}
    Operator $\operatorname{A}: \mathcal{L}^2 \xrightarrow{} \mathcal{L}^2
$ is {\textit{monotone}} if for all $u_1, u_2 \in \dom {\operatorname{A}}$ and $y_i = \operatorname{A}(u_i)$,
\begin{equation}\label{eq:monotonicity}
    \ip{u_1 - u_2}{y_1 - y_2}  \geq 0
\end{equation}
where this definition corresponds to incremental passivity. Operator $\operatorname{A}$ is {\textit{anti-monotone}} if the inner product of (\ref{eq:monotonicity}) is non-positive for all signals in the domain of $\operatorname{A}$. 
\end{definition}

\begin{definition}\label{def:Graph}
   The {\textit{graph}} of an operator $\operatorname{A}: \mathcal{L}^2 \xrightarrow{} \mathcal{L}^2
$ on the space $\mathcal{L}^2$ is defined as
\begin{equation}\label{eq:graph}
    \text{Gra} A = \{(u,y) | y \in \operatorname{A}(u)\} \subseteq \mathcal{L}^2 \times \mathcal{L}^2
\end{equation}
\end{definition}
and the operator $\operatorname{A}$ is \emph{maximal monotone} if it is not contained in the graph of any other monotone operator. 

\begin{definition}\label{def:Resolvent}
    The \emph{resolvent} operator $\operatorname{J}_{\alpha \operatorname{A}}: \mathcal{L}^2 \xrightarrow{} \mathcal{L}^2
 $ of the operator $\operatorname{A}: \mathcal{L}^2 \xrightarrow{} \mathcal{L}^2
 $ is defined as
 \begin{equation}\label{eq:resolvent}
     \operatorname{J}_{\alpha \operatorname{A}} = (\operatorname{Id} + \alpha \operatorname{A})^{-1}
 \end{equation}
 where $\alpha >0$ is a scalar and $\operatorname{Id}$ is the identity operator. 
 \end{definition}

 \begin{remark}\label{rmk:Resolvent-computation}
     The inversion of (\ref{eq:resolvent}) is a relational inverse and always exists \cite{chaffey2023monotone}. Nevertheless, this existence does not guarantee a simple computational procedure. Proximal methods are used to compute this resolvent \cite{parikh2014proximal} and the tractability of this computational stage depends heavily on the structure and properties of the operator $\operatorname{A}$.
 \end{remark}

 \begin{definition}{\label{def:reflected}}
    The \emph{reflected resolvent} operator (\emph{Cayley} operator) $\operatorname{R}_{\alpha \operatorname{A}}: \mathcal{L}^2 \xrightarrow{} \mathcal{L}^2
 $ of the operator $\operatorname{A}: \mathcal{L}^2 \xrightarrow{} \mathcal{L}^2
 $, is defined as
 \begin{equation}
    \operatorname{R}_{\alpha \operatorname{A}} = 2  \operatorname{J}_{\alpha \operatorname{A}} - \operatorname{Id}
 \end{equation}
 where $\operatorname{J}_{\alpha \operatorname{A}}$ is the resolvent operator. An alternative representation of the reflected resolvent is
 \begin{equation}
    \operatorname{R}_{\alpha \operatorname{A}} = (\operatorname{Id} - \operatorname{A}) (\operatorname{Id} + \operatorname{A})^{-1}
 \end{equation}
 that will prove to be useful later in this paper.
 \end{definition}

 \begin{definition}{\label{def:FPI}}
    Given an operator $\operatorname{T}: \mathcal{H} \xrightarrow{}  \mathcal{H}$ and a signal $x \in \mathcal{H}$, a fixed-point iteration, also known as {\textit{Picard}} iteration, is defined as
    \begin{equation}{\label{eq:FPI_def}}
        x^{k+1} = \operatorname{T}(x^k)
    \end{equation}
    where the superscript denotes the iteration count. There are different setups where the fixed-point iteration of (\ref{eq:FPI_def}) converges to $\text{Fix} \operatorname{T}$, the most well-known being
    \begin{enumerate}
        \item $\text{Fix} \operatorname{T}$ is non-empty and $\operatorname{T}$ is a contractive mapping. This is the result of the {\textit{Banach fixed-point Theorem}}.
        \item $\text{Fix} \operatorname{T}$ is non-empty and $\operatorname{T}$ is an averaged operator. Then, the iterations of (\ref{eq:FPI_def}) are known as the {\textit{Krasnosel'skiĭ-Mann}} iteration.
    \end{enumerate}
\end{definition}

\subsection{Model Definition}

The model chosen to represent neurons determines the level of  conformity with neurophysiology, hardware implementability and computational complexity. This work adopts the spiking neuron model from study \cite{ribar2019neuromodulation}. Figure \ref{Fig:Neuron_Architecture} illustrates the architecture of this model which offers three key advantages. First, it is highly modular, allowing components to be added or removed to adjust complexity and behavior. Second, its intuitive structure enables straightforward parameter tuning without requiring expertise in neurodynamics \cite{ribar2020synthesis}. Finally, the model can be implemented using electrical circuits, making it a pragmatic choice for building physical neuromorphic systems \cite{ribar2021neuromorphic}.

\begin{figure}[h]
    \centering
    \includegraphics[width=0.3\textwidth]{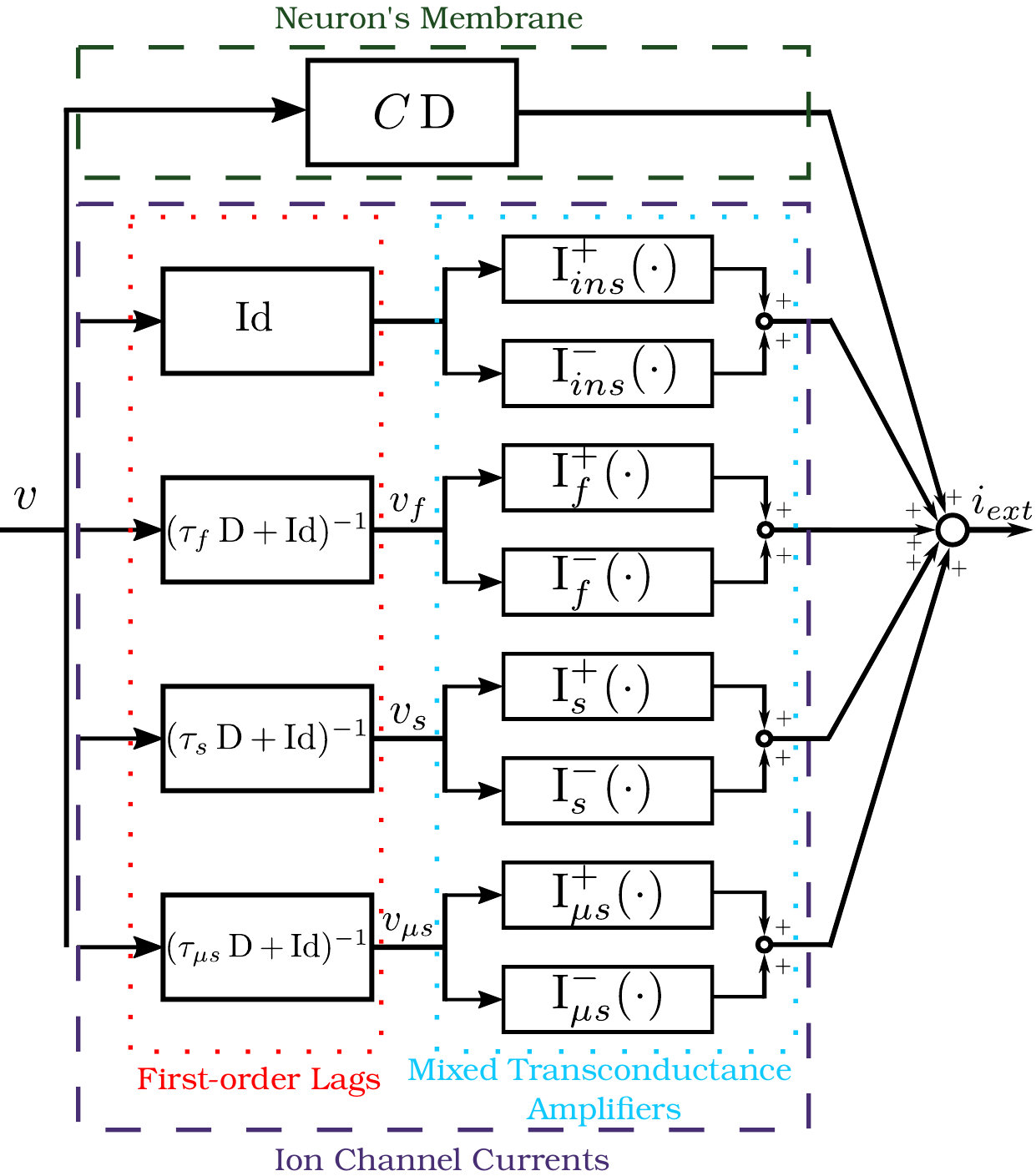} 
    \caption{Modeling the neuron's membrane with a capacitive element is a classical approach in neurodynamics. The first-order lags filter the voltage and generate voltages at different timescales ($v_f, v_s,$ and $v_{\mu s}$). The interaction of positive and negative conductances at different timescales captures the essence of excitability. }
    \label{Fig:Neuron_Architecture}
\end{figure}

In the model of Fig \ref{Fig:Neuron_Architecture}, every element is an operator that acts on signals belonging to $\mathcal{L}^2$. The full dynamics of a single neuron can thus be written as
\begin{equation}\label{Eq:Neuron_Dynamics}
    C \operatorname{D}v + \sum_{x \in \mathbb{T}_x}^{} \operatorname{I}^{\pm}_x ( (\tau _x \operatorname{D} + \operatorname{Id})^{-1}v) = i_{ext}
\end{equation}
where $\operatorname{D}: \mathcal{L}^2 \xrightarrow{}  \mathcal{L}^2 $ denotes the differentiation operator and, $\operatorname{I}^{\pm}_x: \mathcal{L}^2 \xrightarrow{}  \mathcal{L}^2 $ is an abstract representation of either $\operatorname{I}^{+}_x $  or $\operatorname{I}^{-}_x $ where $\operatorname{I}_x^+(v_x) = \alpha ^+ \tanh(v_x - \delta^+ _x)$ is a static nonlinear readout of the filtered voltage and $\alpha ^+ \in \mathbb{R}_{{\geq 0}}$. Similarly, $\alpha ^- \in \mathbb{R}_{{\leq 0}}$ for the operator $\operatorname{I}_x^-(v_x) = \alpha ^- \tanh(v_x - \delta ^- _x)$. The external input signal ${i}_{ext}$ belongs to the same space, meaning ${i}_{ext} \in \mathcal{L}^2$. The interested reader is referred to \cite{ribar2021neuromorphic,ribar2019neuromodulation} for further details and insight on the model and its tuning. $\mathbb{T}_x = \{ins,f,s,\mu s\}$ is the set of existing timescales and contains instantaneous, fast, slow and ultraslow. 

The terms $(\tau _{x} \operatorname{D} + \operatorname{Id} )^{-1}$ are first-order lags that generate voltages at different timescales (e.g., $v_ {\mu s} = (\tau _{\mu s} \operatorname{D} + \operatorname{Id} )^{-1} v$). This operator-theoretic representation of first-order lags corresponds to the ODE $\tau _{\mu s} \dot{v}_{\mu s} = v - v_ {\mu s} $.  

The output-to-input representation of the spiking neuron in Fig \ref{Fig:Neuron_Architecture} intentionally deviates from the traditional input-to-output flow used in classical system theory. In conventional system modeling, one typically applies an input (external current in this case) and analyzes how the system transforms it into an output (membrane voltage in this case). However, in the representation adopted here, this flow is reversed: the membrane voltage $v$, which would traditionally be considered an output, is instead treated as the input to an inverse operator that yields the external current $i_{ext}$.

This inversion significantly simplifies the representation of spiking systems. Specifically, it allows the entire neuron model to be expressed in a modular, \emph{single-layer} architecture. This drastically facilitates the use of splitting algorithms that sit at the center of operator-theoretic solvers. In this configuration, the first component---a capacitive element---models the neuron's membrane, while the remaining parallel branches capture the contributions of various ion channel currents. These branches operate independently and in parallel, filtering the voltage at different timescales and applying nonlinear transformations to produce excitatory or inhibitory current contributions.

The extension from a single neuron to a spiking network, while maintaining the single-layer architecture, is straightforward. Figure \ref{Fig:Network_Architecture} demonstrates this extension. In this figure, the ``neuron" block is exactly the dynamics presented in the left-hand side of (\ref{Eq:Neuron_Dynamics}) that describes the internal dynamics of the neuron. For the $k^{th}$ neuron, the internal ``neuron" operator is denoted as $ \operatorname{N}_k(v_k)$ and is defined as
\begin{equation}\label{Eq:Node_Dynamics}
    \operatorname{N}_k(v_k) = C \operatorname{D}v_k + \sum_{x \in \mathbb{T}_x}^{} \operatorname{I}^{\pm}_x ( (\tau _x \operatorname{D} + \operatorname{Id})^{-1}v_k)
\end{equation}
and the synaptic connections are modeled using the operator $\operatorname{I}_{jk}^{syn}: \mathcal{L}^2 \xrightarrow{}  \mathcal{L}^2 $ where $j$ denotes the index of the pre-synaptic neuron and $k$ is the index of the post-synaptic neuron that receives the current. The formalism of synaptic currents is identical to that of the ion channels in (\ref{Eq:Neuron_Dynamics}) and can be written as
\begin{align}\label{EQ:SynapseOT_Def}
 \operatorname{I}_{jk} ^{syn}( (\tau _{x} \operatorname{D} + \operatorname{Id} )^{-1} v_j) = \alpha _{jk} \operatorname{X}_{sig} (v_{j,x}- \delta^{syn}_{jk})
\end{align}
 which is the mathematical representation of the {\textit{synaptic dynamics}} block of Fig \ref{Fig:Network_Architecture}. Here, $\alpha _{jk}$ is the maximal conductance of the synaptic connection and its sign determines the excitatory or inhibitory nature of the connection. $\delta^{syn}_{jk}$ is an offset parameter used to tune the behavior of the synaptic connection. $\operatorname{X}_{sig}: \mathcal{L}^2 \xrightarrow{}  \mathcal{L}^2 $ represents the Sigmoid operator and $x$ represents the timescale of the voltage of the $j^{th}$ neuron that causes the synaptic current. The sole reason for the change from $\tanh (\cdot)$ in (\ref{Eq:Neuron_Dynamics}) to $\operatorname{X}_{sig}$ in (\ref{EQ:SynapseOT_Def}) is to remain consistent with the literature \cite{ribar2019neuromodulation,ribar2020synthesis}.  In practice, it is straightforward to achieve excitable behaviors solely using $\tanh (\cdot)$, thus creating the entire behavior with purely one format for ion channels.

\begin{figure}[h]
    \centering
    \includegraphics[width=0.3\textwidth]{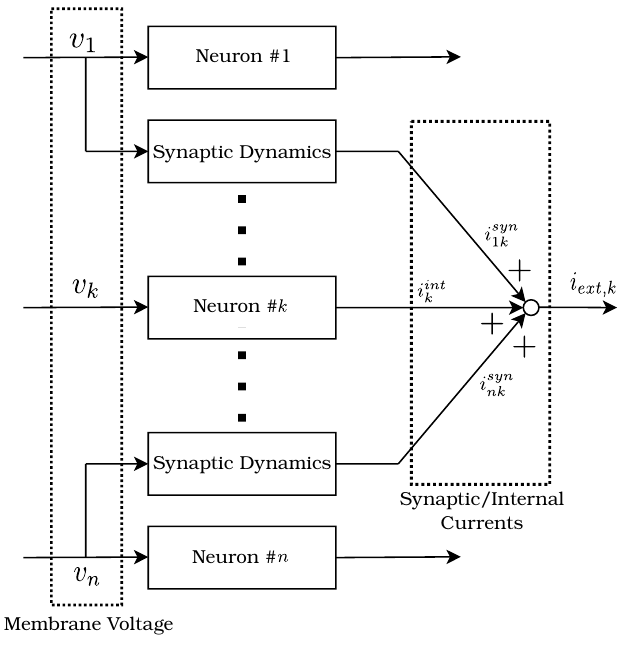} 
    \caption{The spiking network also has the same \emph{single-layer} architecture. It can be seen that the state (voltage) of all neurons can contribute to the dynamics of the $k^{th}$ neuron through synaptic currents. The effect of these synaptic connections is combined with the internal dynamics (denoted by $int$ superscripts) of the $k^{th}$ neuron (stemming from its own ion channels) and equates the externally injected current to the $k^{th}$ neuron.}
    \label{Fig:Network_Architecture}
\end{figure}

Relation (\ref{EQ:SynapseOT_Def}) models {\textit{active}} synaptic connections. Such synaptic couplings can be excitatory or inhibitory meaning that they can actively inject or extract current from neurons. {\textit{Diffusive coupling}} has been discussed and transformed into an operator-theoretic framework in \cite{shahhosseini2024large}.

To have an operator-theoretic description of the spiking network, the individual neurons of (\ref{Eq:Node_Dynamics}) must be combined using the synaptic connections of (\ref{EQ:SynapseOT_Def}) to yield the structure provided by Fig \ref{Fig:Network_Architecture} that leads to 
\begin{equation}{\label{EQ:network}}
    \begin{bmatrix}
    \operatorname{N}_1(\cdot) & \operatorname{I}_{21}^{syn}(\cdot)  & \cdots & \operatorname{I}_{n1}^{syn}(\cdot)\\ 
    \operatorname{I}_{12}^{syn}(\cdot) & \operatorname{N}_2(\cdot) &  & \operatorname{I}_{n2}^{syn}(\cdot)\\ 
    \vdots &  & \ddots & \vdots\\ 
    \operatorname{I}_{1n}^{syn}(\cdot) & \operatorname{I}_{2n}^{syn}(\cdot) &  \cdots & \operatorname{N}_n(\cdot)
    \end{bmatrix} \begin{bmatrix}
    v_1\\ 
    v_2\\ 
    \vdots\\ 
    v_n
    \end{bmatrix} = 
    \begin{bmatrix}
        i_{ext,1}\\ 
        i_{ext,2}\\ 
        \vdots\\ 
        i_{ext,n}
     \end{bmatrix}.
\end{equation}

This representation can be more compactly written as
\begin{equation}{\label{EQ:compact}}
    \operatorname{W}(\boldsymbol{{v}}) = \boldsymbol{i}_{ext}
\end{equation}
where $\operatorname{W}: \boldsymbol{\mathcal{L}}^2 \rightarrow \boldsymbol{\mathcal{L}}^2$ is the operator defining the dynamics of the spiking network and $\boldsymbol{{v}} = ( v_1, \dots, v_n)$ is the augmented membrane voltage vector that consists of all the membrane voltages, as denoted in (\ref{EQ:network}), and  $\boldsymbol{{v}} \in \boldsymbol{\mathcal{L}}^2$. This is similar for $\boldsymbol{i}_{ext}$.

In this representation, any neuron can have an arbitrary synaptic connection with another neuron, with no restrictions on network topology. However, specific network topologies may be of particular interest, such as the symmetric connections observed in Hopfield networks \cite{hopfield1982neural}. Thus, the problem of simulating spiking networks reduces to finding the augmented membrane voltage that satisfies (\ref{EQ:compact}).

\section{Methodology}

\subsection{Problem Formulation}

Equation (\ref{EQ:compact}) describes the dynamics of the spiking network. In fact, for a fixed set of external currents ($\boldsymbol{i}_{ext}^*$), the problem of simulating the behavior is equivalent to finding the set of membrane voltages $\boldsymbol{{v}}$ that satisfies 
\begin{equation}{\label{EQ:ZFP}}
    \operatorname{W}(\boldsymbol{{v}}) - \boldsymbol{i}_{ext}^* = \mathbf{0}
\end{equation}

This means that the problem of simulating the spiking network is in fact merely a zero-finding problem (ZFP) in the lifted Hilbert space; to find what  $\boldsymbol{{v}}$ satisfies (\ref{EQ:ZFP}). 

\subsubsection{Monotonicity---The Property of Interest}

Given the structure and properties of the operator $\operatorname{W}(\cdot)$, different methods can be used to solve (\ref{EQ:ZFP}). Among these properties, {\emph{monotonicity}} holds a special place. This property bridges physical interpretability with computational tractability. From a physical perspective, it indicates the incremental passivity of the elements or subsystems that it represents. Monotonicity is also closely related to convex optimization methods since the subgradient of any convex function is monotone.  Consequently, computationally tractable and efficient large-scale methods can be built for this class of systems \cite{rockafellar2015convex}. First-order methods that have recently drawn extensive attention rely on this exact property \cite{ryu2022large}.

The most common technique for solving monotone ZFPs is to convert them into fixed-point problems and then solve those fixed-point problems using fixed-point iteration (FPI) algorithms. In fact, monotonicity is the property that ensures the convergence of the fixed-point iteration algorithms toward a solution and links operator theory with convex optimization. This link, in its most basic form is discussed in the following theorem.

\begin{theorem}{\label{ZPFtoFPI}} 
    If $\operatorname{A}$ is maximally monotone and the set of the solutions of $\text{Zer} \operatorname{A}$ is non-empty, then
    \begin{equation}
        \text{Zer} \operatorname{A} = \text{Fix} \operatorname{J}_{\alpha \operatorname{A}}
    \end{equation}
    where $\operatorname{J}_{\alpha \operatorname{A}}$ is the resolvent operator of $\operatorname{A}$. To find the solutions of $\text{Fix} \operatorname{J}_{\alpha \operatorname{A}}$, fixed-point iteration algorithms (e.g., Proximal Point Algorithm) can be used \cite{bauschke_convex_2011}.
\end{theorem}

Study \cite{chaffey2023monotone} simulates monotone one-port circuits using operator-theoretic methods and demonstrates their advantages.

\subsubsection{Beyond Simple Monotone Problems}

The dynamics of the spiking neuron (\ref{Eq:Node_Dynamics}) and network (\ref{EQ:network}) can never be monotone. The justification for the breakdown of monotonicity is threefold. 

\begin{enumerate}
    \item \textbf{Nature of Excitability}: Excitability and spiking behavior is grounded in the interaction of active elements at fast timescales with dissipative elements at slower timescales \cite{ribar2021neuromorphic}. The active elements inject energy to the system and locally destabilize it while dissipative elements restore the equilibrium. By essence, the operator-theoretic description of this mechanism results in a mixed-monotone representation of spiking systems that is aligned with the mixed-feedback structure of spiking system \cite{sepulchre2022spiking}.
    \item \textbf{Model Component Structure}: The model of both internal ion channels (\ref{Eq:Neuron_Dynamics}) and synaptic connections (\ref{EQ:SynapseOT_Def}) is constructed via the composition of a static nonlinear monotone readout with a first-order lag (that is also monotone). Counterintuitively, monotone operators are not closed under composition and this causes another breakdown of monotonicity.
    \item \textbf{Off-diagonal Synaptic Dynamics}: The synaptic currents in (\ref{EQ:network}) are off-diagonal element of the network operator and can easily break monotonicity. This can be easily verified by considering a 2-by-2 network and plugging the network operator into the definition of monotonicity, as decribed in Definition \ref{def:Monotonicity}.
\end{enumerate}

Nevertheless, the loss of monotonicity does not correspond to a full loss of structure and does not transform the problem into an arbitrary ZFP. The modularity and architecture of the model used in this paper allow the decomposition of (\ref{EQ:compact}) into monotone and anti-monotone parts and the problem can be reformulated to a {\textit{difference of monotone}} (DM) zero finding problem. This reformulation grants the opportunity to borrow methods from difference of convex (DC) optimization and adapt them to efficient first-order methods grounded in operator theory to solve (\ref{EQ:compact}). In fact, a conventional method to solve nonconvex (with DC being a trivial subset) problems is to exploit Majorize-Minimization (MM) algorithms \cite{ortega2000iterative}. As it will become evident in the following section, the first-order methods that solve DM and DC problems are instances of the MM algorithms that are made tractable due to the structured formalism of the problem.

It is important to note that even if $\operatorname{W}$ were monotone, directly applying Theorem \ref{ZPFtoFPI} would still be impractical. The bottleneck lies in computing the resolvent of $\operatorname{W}$, which is often computationally intractable for sophisticated operators. This issue is a well-known limitation in fixed-point iteration methods. To overcome this, \textit{splitting methods} were introduced. These methods decompose a complex operator—such as $\operatorname{W}$—into a sum of simpler operators and modify the FPI scheme so that it relies only on the resolvents of the simpler operators. This divide-and-conquer approach is central to many scalable and distributed optimization algorithms. One of the most prominent examples in the control and optimization community is the Alternating Direction Method of Multipliers (ADMM) \cite{boyd2011distributed}.

\begin{remark}{\label{splitting_properties}}
    It is trivial that the splitting of the original operator is not arbitrary and the smaller operators must bear certain properties. First and foremost, the resolvent of the smaller operators should be numerically cheap to compute. Second, the smaller operators must be structured and entail specific properties (such as monotonicity) to ensure the convergence of the FPI acting on the split formalism. Finally, it would be ideal if the smaller operators provided physical interpretability as to how the smaller operators contribute to the grand scheme of the ZFP. Thus, the next immediate challenge is to find a natural splitting for the spiking neural network representation of (\ref{EQ:network}).
\end{remark}

\subsection{Splitting the Spiking Neuron}{\label{sec:shifts}}

To be able to solve the ZFP of (\ref{EQ:ZFP}), the problem has to be split into smaller operators while maintaining the difference of monotone structure. The single-layer architecture of the spiking neuron and the network will be central in obtaining the most natural splitting.

We first consider the case of an isolated spiking neuron. The objective is to break (\ref{Eq:Node_Dynamics}) into the difference of smaller monotone operators. Inspired by the single-layer architecture of the spiking neuron, illustrated in Fig. (\ref{Fig:Neuron_Architecture}), the splitting is of the form
\begin{equation}{\label{EQ:full}}
    \operatorname{N(\cdot)} = C \operatorname{D(\cdot)} + \sum_{x \in \mathbb{T}_x}^{} (\operatorname{A}_x(\cdot) - \operatorname{B}_x(\cdot))
\end{equation}
where $\operatorname{A}_{x}(\cdot) = \operatorname{\operatorname{I}_{x} ^+((\tau _{x} \operatorname{D} + \operatorname{Id} )^{-1}(\cdot))}$ that corresponds to positive conductance at timescale $x$ while $-\operatorname{B}_{x}(\cdot) = \operatorname{\operatorname{I}_{x} ^-((\tau _{x} \operatorname{D} + \operatorname{Id} )^{-1}(\cdot))}$ represents the negative conductance at timescale $x$. With this splitting, each smaller operator ($\operatorname{D}$, $\operatorname{A}_{x}$ or $\operatorname{B}_{x}$) corresponds to one element of the single-layer architecture. Furthermore, these smaller operators bear physical interpretation and signify the either the capacitor or the ion channel currents. 

Although $\operatorname{A}_{x}$ and $\operatorname{B}_{x}$ are not monotone, they are {\textit{hypomonotone}} \cite{combettes2004proximal}, meaning that they can be made monotone through simple shifts. The difference of monotone structure of the problem permits the addition and subtraction of a shift $\lambda \times \operatorname{Id}$ to $\operatorname{A}_{x}$ and  $\operatorname{B}_{x}$ (where $\lambda > 0$). The difference format of (\ref{EQ:full}) cancels out this addition, so the original problem remains unchanged, as illustrated in Figure \ref{fig:hypomonotone}.

\begin{figure}[h]
    \centering
    \includegraphics[width=0.4\textwidth]{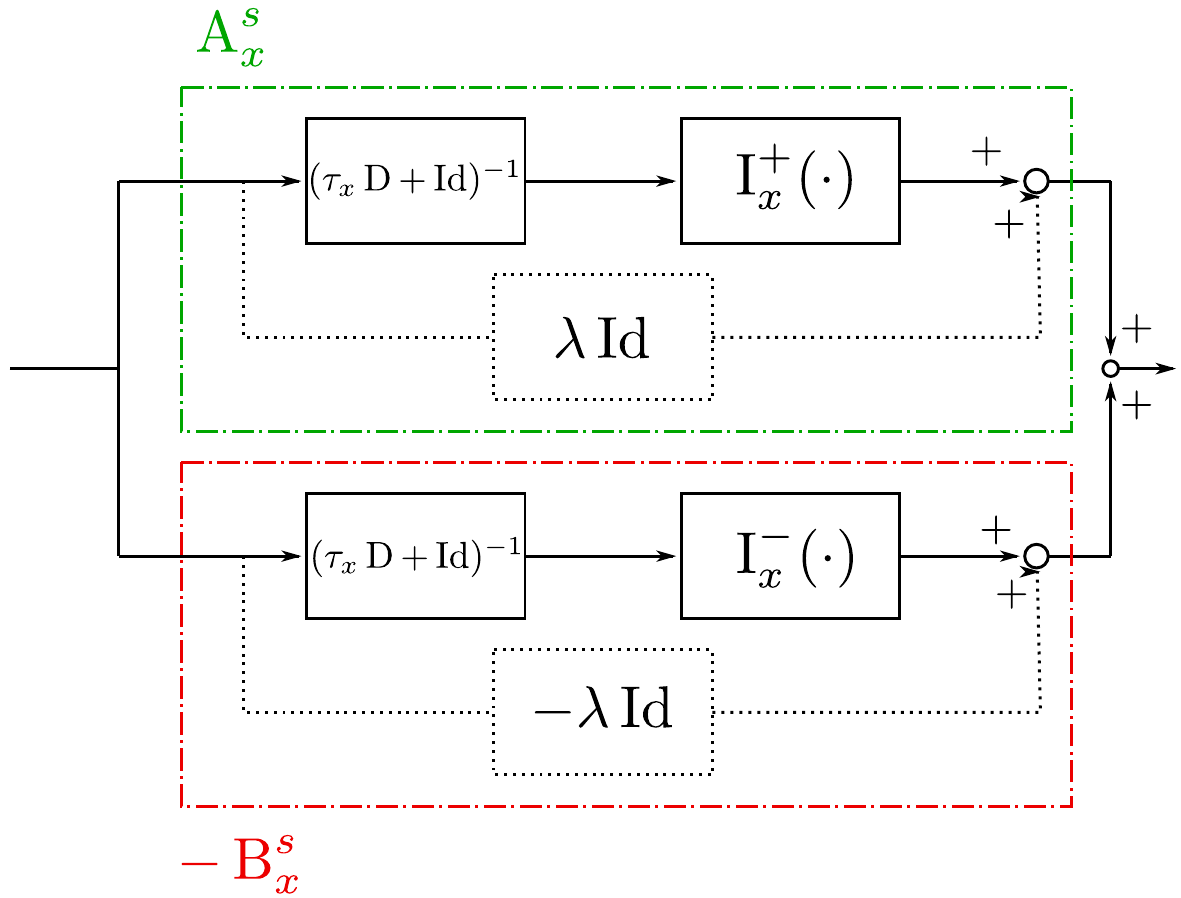} 
    \caption{Addition of a positive shift operator to the composition of monotone operators makes the upper block monotone. The same shift is subtracted from the composition of the anti-monotone operator and the monotone operator to make it anti-monotone. This virtual addition does not alter the dynamics of the system.}
    \label{fig:hypomonotone}
\end{figure}

Physically, this corresponds to adding a linear resistor with positive conductance to the $\operatorname{I}^+_{x}$ branch and adding the exact same linear resistor but with negative conductance to the $\operatorname{I}^-_{x}$ branch of the model of Fig. \ref{Fig:Neuron_Architecture}. This yields a difference-of-monotone decomposition of the nodal dynamics with physical interpretability. The new shifted operators $\operatorname{A}_x^{\text{s}}$ and $\operatorname{B}_x^{\text{s}}$ are defined as
\begin{equation}{\label{EQ:shiftedOperators}}
   \operatorname{A} _{x}^{\text{s}} = \operatorname{A}_{x} + \lambda \operatorname{Id}, \quad \operatorname{B}_{x}^{\text{s}} = \operatorname{B}_{x} +\lambda \operatorname{Id}
\end{equation}
and so the splitting of the spiking neuron can be written as
\begin{equation}{\label{EQ:node_final}}
    \operatorname{N(\cdot)} = C \operatorname{D(\cdot)} + \sum_{x \in \mathbb{T}_x}^{} (\operatorname{A} _{x}^{\text{s}}(\cdot) - \operatorname{B} _{x}^{\text{s}}(\cdot))
\end{equation}
where each smaller operator is monotone.

\subsection{Splitting the Spiking Network}

The single-layer architecture of the spiking network, demonstrated in Fig \ref{Fig:Network_Architecture}, can also be used to obtain a network-level splitting. The splitting begins with decomposing (\ref{EQ:network}) into two subsystems, one taking into account only the diagonal elements (corresponding to neuron block dynamics) and the other covering off-diagonal elements (corresponding to synaptic dynamics) as

\begin{equation}
    \operatorname{W} = \operatorname{N}_{\text{full}}+ \operatorname{Y}.
\end{equation}

The operator $\operatorname{N}_{\text{full}}= diag(\operatorname{N}_1,\cdots,\operatorname{N}_n)$  is diagonal and as a consequence, each neuron block ($\operatorname{N}_k$) is isolated and its resolvent can be computed separately. It is obvious that the splitting of (\ref{EQ:node_final}) can now be used to break the $\operatorname{N}_{{\text{full}}}$ into smaller operators. 

The next step is to decompose the operator that accounts for synaptic couplings and is the off-diagonal part of (\ref{EQ:network}), denoted as $\operatorname{Y}$ and named the {\textit{synaptic dynamics}}. It is possible to split this operator into $n$ smaller operators where each denote the contribution of the synaptic connections to the dynamics of the $k^{th}$ neuron, as illustrated in Fig \ref{Fig:Synaptic_Architecture}. The interpretation is that the aggregate effect of all the synaptic currents on the $k^{th}$ neuron corresponds to the $k^{th}$ row of synaptic dynamics operator $\operatorname{Y}$ and is denoted as $\operatorname{Y}_k$. 

\begin{figure}[h]
    \centering
    \includegraphics[width=0.4\textwidth]{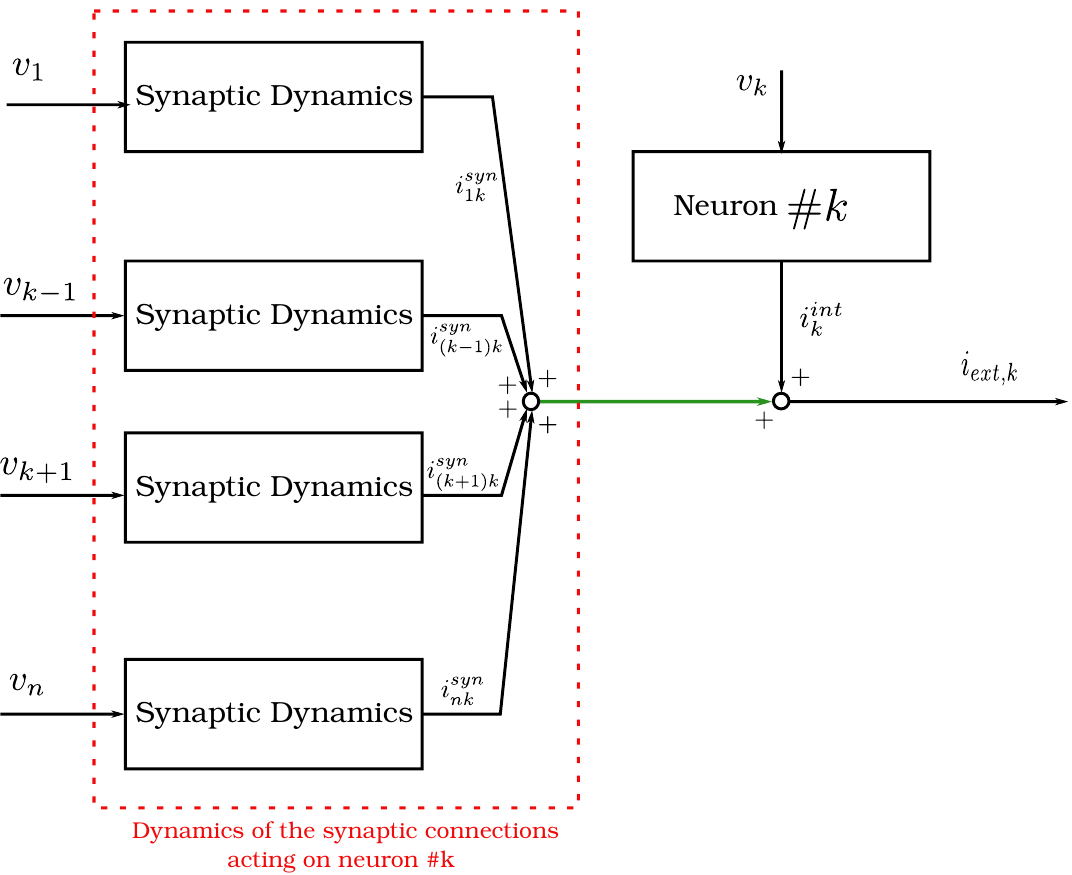} 
    \caption{The effect of the all the synaptic connections of the network on the $k^{th}$ neuron is demonstrated here. The green link represents the aggregate current that the $k^{th}$ neuron receives from its synaptic connections. Given the active nature of the synaptic connections, they rely upon the voltage of the pre-synaptic neuron. The combination of the externally injected current, synaptic currents and internal dynamics' currents govern the behavior of the neuron.}
    \label{Fig:Synaptic_Architecture}
\end{figure}

An issue with the smaller operator $\operatorname{Y}_k$ is its mixed monotonicity. With an argument similar to the one used for $\operatorname{A}_x$, it is possible to make these operators monotone by adding and subtracting the identity operator. Thus, it is possible to decompose $\operatorname{Y}$ as
\begin{equation}{\label{eq:synaptic_splitting}}
    \operatorname{Y} =  \sum_{i=1}^{n}\left ( \underset{\text{monotone}}{\underbrace{(\lambda _i ^{\text{syn}}\operatorname{Id} + \operatorname{Y}_i)}} -\underset{\text{monotone}}{\underbrace{\lambda _i ^{\text{syn}}\operatorname{Id}}} \right )
\end{equation}
where again, $\operatorname{Y}_i$ denotes an operator that contains the $i^{th}$ row of the $\operatorname{Y}$ operator and zeros everywhere else. For simplicity, we define the new operator $\operatorname{C}_i = \lambda _i ^{\text{syn}}\operatorname{Id} + \operatorname{Y}_i$ with $\operatorname{C}_i : \boldsymbol{\mathcal{L}}^2 \rightarrow  \boldsymbol{\mathcal{L}}^2$. Now, operator $\operatorname{Y}$ is also decomposed into a difference of monotone structure. Mathematically, this corresponds to
\begin{equation}{\label{EQ:Split_synaptic}}
    \operatorname{Y} =  \sum_{i=1}^{n}\left ( \underset{\text{monotone}}{\underbrace{\operatorname{C}_i}} -\underset{\text{monotone}}{\underbrace{\lambda _i ^{\text{syn}} \operatorname{Id}}} \right )
\end{equation}

It must be noted that the resolvent of operator $\operatorname{C}_i$ is extremely easy to compute via elementary linear algebra operations. It is now possible to combine the splitting in (\ref{EQ:node_final}) and (\ref{EQ:Split_synaptic}) to provide the full splitting of a spiking neural network. Explicitly, for a network with $n$ neurons
\begin{align}\label{eq:full_splitting}
    & \operatorname{W} = \operatorname{N}_{\text{full}} + \operatorname{Y} = \\ \nonumber
    & \underset{\text{Effect of Neuron Dynamics}}{\underbrace{diag(\cdots,C_k \operatorname{D} + \sum_{x \in \mathbb{T}_x^k}^{} ({\operatorname{A}} _{x,k}^{s}(\cdot) - {\operatorname{B}} _{x,k}^{s}(\cdot)),\cdots)}}   \\
    & +\underset{\text{Effect of Synaptic Dynamics}}{\underbrace{\sum_{i=1}^{n}\left ( \operatorname{C}_i -\lambda _i ^{\text{syn}} \operatorname{Id} \right )}} \nonumber
\end{align}
will be split into $2n+2m+1$ smaller operators, where $m$ is the number of timescales. The new subscript $k$ used in (\ref{eq:full_splitting}) refers to the $k^{th}$ neuron, as the internal dynamics of the neurons need not be identical. This splitting is aligned with the requirements discussed in remark \ref{splitting_properties}. The operator describing the spiking network $\operatorname{W}$ is broken into smaller operators that are computationally easy to deal with. The splitting has a clear difference of monotone structure and retains the physical interpretation of the circuit topology.

\subsection{A Novel Splitting Algorithm to Solve Spiking Networks}

Splitting algorithms for convex optimization and monotone operators have gained popularity for their scalability. Their success stems from their strong ties to convex optimization and the contraction properties of the resolvent operator. Prominent examples include the Alternating Direction Method of Multipliers (ADMM) \cite{boyd2011distributed}, Forward-Backward Splitting (FBS) \cite{ryu2022large}, and Davis-Yin Splitting (DYS) algorithm \cite{davis2017three}.

In comparison, splitting algorithms for DM problems are less developed. This is because DM problems are extremely broad, covering nearly all possible problems. 

Specialized DM algorithms were developed recently \cite{chuang2022unified}. Nevertheless, these methods can only handle up to three operators, which is limiting. To address this limitation, we propose a novel splitting algorithm that can handle any number of operators in a DM setup.

\subsubsection{Consensus-based Difference of Monotone Douglas-Rachford Algorithm}

The problem of interest has the formulation
\begin{equation}{\label{eq:method}}
    \operatorname{E}(x) + \sum\limits_{i=1}^{p}(\operatorname{F}_i(x)-\operatorname{G}_i(x)) = 0
\end{equation}
where $\operatorname{E}$, $\operatorname{F}_i$ and, $\operatorname{G}_i$ are monotone operators. To further exploit the structure of our problem and simplify the proof of convergence, further assumptions are put on the operators. The signal space of (\ref{eq:method}) is constrained to only include energy-bounded signals, thus $x \in \mathcal{L}^2$. Further, set $I = \{1,\ldots, p\}$ and
    \begin{equation}{\label{eq:definitions}}
    \left\{
    \begin{array}{l}
    \boldsymbol{C} = \{(z, \ldots, z) \in \boldsymbol{\mathcal{L}}^2 \mid z \in \mathcal{L}^2\}, \\
    {\operatorname{j}}: \mathcal{L}^2 \to \boldsymbol{C}, \; x \mapsto (x, \ldots, x), \\
    \operatorname{P}_{\boldsymbol{C}} : \boldsymbol{\mathcal{L}}^2 \to \boldsymbol{\mathcal{L}}^2, \operatorname{P}_{\boldsymbol{C}} \boldsymbol{z}={\operatorname{j}}\left(\frac{1}{p} \sum_{i \in I} z_{i}\right) \\
    \boldsymbol{\operatorname{F}}: \boldsymbol{\mathcal{L}}^2 \to \boldsymbol{\mathcal{L}}^2, x \mapsto (F_1(x_1), F_2(x_2), \ldots, F_p(x_p))\\
    \boldsymbol{\operatorname{G}}: \boldsymbol{\mathcal{L}}^2 \to \boldsymbol{\mathcal{L}}^2, x \mapsto (G_1(x_1), G_2(x_2), \ldots, G_p(x_p))\\
    \end{array}
    \right.
    \end{equation}

    where $\boldsymbol{\mathcal{L}}^2$ denotes the lifted space described in the mathematical preliminaries section, $\boldsymbol{C}$ represents the consensus set, ${\operatorname{j}}$ describes the lifting operator, $\operatorname{P}_{\boldsymbol{C}}$ is the operator that projects a signal  to the consensus set $\boldsymbol{C}$. $ \boldsymbol{\operatorname{F}}$ and $ \boldsymbol{\operatorname{G}}$ are concatenations of the operators $\operatorname{F}_i$ and $\operatorname{G}_i$ respectively, and these definitions closely match that of the literature \cite{bauschke_convex_2011}.

    \begin{theorem}{\label{thm:main}}
    Let $\mathcal{D} \subseteq \mathcal{L}$ be closed, convex and nonempty, and $\alpha > 0$.  Let $\operatorname{E}, \operatorname{F}_i, \operatorname{G}_i,: \mathcal{L}^2 \to \mathcal{L}^2$ satisfy the following assumptions:
    \begin{enumerate}
        \item Operator $\operatorname{E}$ is $\rho$-strongly monotone on $\mathcal{D}$
        \item Operator  $\operatorname{F}_i$ is $\gamma$-strongly monotone on $\mathcal{D}$ for all $i \in I$
        \item Operator $\operatorname{G}_i$ satisfies
        \begin{equation}
             (\beta + \frac{1}{\epsilon})\norm{y - \hat{y}}^2 \geq \ip{\operatorname{G}_i y - \operatorname{G}_i \hat{y}}{y-\hat{y}} \geq \beta \norm{y - \hat{y}}^2
        \end{equation}
        for all $y, \hat y \in \mathcal{D}$ and $i \in I$.
        \item The condition $\gamma > \beta + 1 + \frac{2}{\epsilon}$ holds on the subset ${\mathcal{D}}$
        \item $\boldsymbol{\operatorname{T}}: \boldsymbol{\mathcal{L}}^2 \to \boldsymbol{\mathcal{L}}^2$ defined by 
    \begin{equation}{\label{eq:FPI}}
        \begin{aligned}
        & \boldsymbol{\operatorname{T}}(\cdot) = \operatorname{Id}(\cdot) - \operatorname{J}_{\alpha \boldsymbol{\operatorname{E}}}(\operatorname{P}_{\boldsymbol{C}}(\cdot)) \\
        & + \operatorname{J}_{p \alpha \boldsymbol{\operatorname{F}}}(2\operatorname{J}_{\alpha \boldsymbol{\operatorname{E}}}(\operatorname{P}_{\boldsymbol{C}}(\cdot)) - \operatorname{Id}(\cdot) + p \alpha \boldsymbol{\operatorname{G}}(\operatorname{J}_{\alpha \boldsymbol{\operatorname{E}}}(\operatorname{P}_{\boldsymbol{C}}(\cdot)))
        \end{aligned}
    \end{equation}
    satisfies $T(\boldsymbol{\mathcal{D}}) \subseteq \boldsymbol{\mathcal{D}}$,
    \end{enumerate}
    where $\boldsymbol{\mathcal{D}} = 
    \bigoplus_{i \in I} \mathcal{D}$ is the lifted subset.
    Then, the fixed-point iteration
    \begin{equation}{\label{eq:FPI_basic}}
        \boldsymbol{z}^{k+1} = \boldsymbol{\operatorname{T}}(\boldsymbol{z}^k)
    \end{equation}
    converges to a solution of \eqref{eq:method}, if such a solution exists.

    \begin{proof}: The proof of convergence for this FPI algorithm is discussed in detail in Appendix \ref{appendix}. 
    \end{proof}
\end{theorem}

The result of Theorem \ref{thm:main} offers an FPI algorithm that can solve the DM problem of (\ref{eq:method}). As previously discussed, the condition of strong monotonicity is not a strong imposition since the difference of monotone structure allows the addition and subtraction of elements without affecting the original problem. The pseudo-code for this algorithm is provided as follows.

\begin{algorithm}[H]
    \caption{Consensus-based Difference of Monotone Douglas-Rachford}
    \begin{algorithmic}[1]
    \ENSURE  Satisfaction of assumptions related to $\operatorname{E}$, $\operatorname{F}_i$ and, $\operatorname{G}_i$
    \\ \textit{Initialisation} : Initial guesses for $x$ and $z$
    \\ 
     \FOR {$k = 1$ to $\text{max-iteration}$}
     \STATE $x^{k+1} = \operatorname{J}_{\alpha \operatorname{E}}(\bar{z}^k)$
     \FOR{$i=1:m$}
     \STATE $z^{k+1} _i = z^{k} _i - x^{k+1} + \operatorname{J}_{p \alpha \operatorname{F}_i}(2x^{k+1} - z^{k} _i +p \alpha \operatorname{G}_i(x^{k+1}))$
     \ENDFOR
     \STATE $\bar{z}^{k+1} = \frac{1}{p} \sum_{i = 1}^{p} z_i^k$
     \ENDFOR
    \RETURN $x$
    \end{algorithmic}
\end{algorithm}

With this algorithm, it is now possible to solve arbitrary spiking neural networks after recasting them into the form (\ref{eq:full_splitting}) and then, (\ref{eq:method}).

\subsubsection{Technical Notes on Efficient Computational Methods}

Although Algorithm 1 uses first-order methods, this does not guarantee computational efficiency. The main bottleneck lies in computing the resolvent operator $\operatorname{J}_{ \operatorname{F}_i}$. Due to the nonlinear and dynamic nature of the $\operatorname{F}_i$ operators, a closed-form solution for their resolvent is not feasible. However, numerical methods can approximate the output of this operator. In this section, a computational method is developed to numerically approximate $\operatorname{J}_{ \operatorname{F}_i}$ using an inner iterative solver.

\begin{remark}
    Given the single-layer architecture of the network and the rigid structure of each line in this single-layer architecture, every line (except for the capacitor) can be reformulated to have the form

    \begin{equation}{\label{EQ:Ax}}
        \operatorname{F}_i(\cdot) = \alpha ^+ \tanh((\tau _{x} \operatorname{D} + \operatorname{Id} )^{-1} (\cdot) - \delta _x)
    \end{equation}
    and although the synaptic dynamics uses a Sigmoid nonlinearity, the overall structure is identical. Thus, any numerical procedure developed for (\ref{EQ:Ax}) can be used for the synaptic dynamics as well.

\end{remark}

  Using Definition \ref{def:Resolvent}, it is possible to write the resolvent of $\operatorname{F}_i$ as
    \begin{equation}{\label{EQ:AxR}}
        \operatorname{J}_{\gamma \operatorname{F}_i}(\cdot) = (\operatorname{Id} + \gamma \operatorname{F}_i(\cdot))^{-1}
    \end{equation}
and then apply this operator on a known signal $w$ to obtain the signal $q$. This results in the expression $q = \operatorname{J}_{\gamma \operatorname{A}_x}(w)$. By using (\ref{EQ:Ax}) and (\ref{EQ:AxR}), it is possible to obtain
    \begin{equation}{\label{EQ:dynamicSplit}}
        w  = q + \gamma \operatorname{F}_i(q) = q + \gamma \alpha ^+ \tanh((\tau _{x} \operatorname{D} + \operatorname{Id} )^{-1} (q) - \delta _x)
    \end{equation}
   where the auxiliary operator $\operatorname{L}$ and auxiliary variable $u$  are defined as $\operatorname{L}= (\tau _{x} \operatorname{D} + \operatorname{Id} )$ and $u = \operatorname{L}^{-1}q$. It is now possible to rewrite (\ref{EQ:dynamicSplit}) as
    \begin{equation}{\label{EQ:dynsplit}}
        \operatorname{L}u + \gamma \alpha ^+ \tanh(u - \delta _x) = w.
    \end{equation}
    
    In this setup, given  a known signal $w$, (\ref{EQ:dynsplit}) is a ZFP that consists of two monotone operators. Methods such as Douglas-Rachford splitting algorithms can be used to solve this problem \cite{ryu2022large}. These methods work for all choices of $\gamma>0$, but one needs to solve this ZFP for every timescale of every neuron and every non-instantaneous synaptic connection. Thus, this step makes up a major computational portion of the algorithm and must be done as efficiently as possible. In fact, the use of methods such as Forward-Backward or Douglas-Rachford splitting is computationally expensive, even with a warm-start. This motivates introducing a computationally cheap approach for evaluating this resolvent. 

As discussed in \ref{sec:shifts}, operator $\operatorname{F}_i$ can be shifted to conform to the monotonicity requirements of splitting methods. After a shift, this operator can be written as
 \begin{equation}{\label{EQ:ShiftedFancy}}
        \operatorname{F}_i^{\text{s}}(\cdot) = \alpha ^+ \tanh((\tau _{x} \operatorname{D} + \operatorname{Id} )^{-1} (\cdot) - \delta _x) +\lambda \operatorname{Id}
    \end{equation}
and this new operator leads to the resolvent expression of
\begin{equation}
    q +\gamma  \lambda  q + \gamma \alpha ^+ \tanh(\operatorname{L}^{-1} (q) - \delta _x) = w
\end{equation}
and this expression can directly be transformed into a fixed-point iteration 
  \begin{equation}{\label{eq:AlternativeResolventShifted}}
        q^{o+1}  = \frac{w}{1 + \lambda \gamma} - \frac{\gamma \alpha ^+}{1 + \lambda \gamma} \tanh((\tau _{x} \operatorname{D} + \operatorname{Id} )^{-1} (q^o) - \delta _x)
    \end{equation}
where $o$ denotes the counter for this iterative solver that resides within the main iterative solver of Algorithm 1. The last remaining step is to show the conditions for the convergence of this numerical scheme. This is done through showing the contractiveness of the operator governing the FPI of (\ref{eq:AlternativeResolventShifted}), explicitly formulated as
\begin{equation}{\label{EQ:AlternativeResolvent}}
    \operatorname{T}_{\text{res}}(q) =  \frac{w}{1 + \lambda \gamma} - \frac{\gamma \alpha ^+}{1 + \lambda \gamma} \tanh((\tau _{x} \operatorname{D} + \operatorname{Id} )^{-1} (q) - \delta _x)
\end{equation}

To find the conditions that ensure this contractiveness, the Lipschitz continuity of the $\tanh$ function, described as $\norm{\tanh(x) - \tanh(y)} \leq \norm{x - y}$ is exploited. To obtain the conditions of contractiveness, it is possible to use the definition OF contractiveness and obtain the expression 
\begin{equation}{\label{limitsofshift}}
    \begin{aligned}
        \norm{ \operatorname{T}_{\text{res}}(q^{o+1}) - \operatorname{T}_{\text{res}}(q^{o})} \leq |\frac{\gamma \alpha ^+}{1 + \lambda \gamma}|  \norm{q^{o+1} - q^{o}}
    \end{aligned}
\end{equation}
where, by the Banach fixed-point theorem, this iterative scheme converges upon having a gain of less than one for the operator. Thus, the condition reduces to $|\frac{\gamma \alpha ^+}{1 + \lambda \gamma}| < 1$. In this setup, it is possible to choose large values of $\gamma$ and quickly compute the resolvent by choosing an adequate offset $\lambda$. This is not feasible without introducing a shift since without a shift, the contractiveness condition reduce to $|\gamma \alpha ^+| < 1$.

\subsubsection{Time-Frequency Hopping and Exploiting the Structure}

Given the structured architecture of the model, it is possible to further accelerate the computational steps. The dynamic operators, including both $\operatorname{D}$ that represents the capacitor and $\operatorname{L}$ that represents the first-order lags, are LTI systems. The goal of this subsection is to further exploit this structure. 

It is well-known that LTI systems are easier to deal with in the frequency domain since their representation will be diagonal. The idea of time-frequency hopping is to perform the computations of operators that consist solely of LTI in the frequency domain and the computations of static nonlinearities in the time domain. This hopping between the frequency domain and the time domain is achieved by computing the FFT and iFFT of the signals. The details of this approach can be found in \cite{shahhosseini2024large}. Exploiting the structure of LTI operators ensures that the computational costs of this step remain minimal and of the order $\mathcal{O}(g \log g)$ where $g$ is the length of the discretization of vector $v_k$. This means that line 2 of Algorithm 1 which corresponds to the $\operatorname{D}$ is carried out in the frequency domain while line 4 is carried out in the time domain.

\section{Illustrations and Discussions}

This section presents three examples to demonstrate the capabilities of the proposed method. It begins with a simple example as a tutorial and demonstrates the built-in mechanisms for variability analysis. It then proceeds to simulate a bursting neuron to showcase the solution evolution mechanism for of this approach. Afterward, a half-center oscillator system \cite{brown1914nature}, consisting of two bursting neurons, is simulated and the {\emph{behavior-informed}} nature of this solver, alongside its {\emph{event-capturing}} mechanism, is shown.

\subsection{Solver Configurations and Implications}

The proper implementation of this simulation framework allows the efficient use of computational resources. This subsection explores four different aspects that are directly involved with the technical implementation of this framework. It provides intuition as to how they affect the simulation and how they should be chosen or tuned.

\subsubsection{Pre-processing frequent components}

Due to the LTI structure of the operator $\operatorname{E}$, it is possible to construct its matrix representation and consequently, its resolvent $\operatorname{J}_{\alpha \operatorname{E}}$ in the frequency domain. This matrix is used at every iteration of Algorithm 1 and does not change throughout the simulation. Therefore, it is computationally convenient to compute this term once at initialization and repeatedly use it within the structure of the solver. The same argument can be used for the first-order lags within the structure of $\operatorname{F}_i$ and $\operatorname{G}_i$. The pre-computation of these operators permits their efficient call during the simulation.

\subsubsection{Choice of the hyperparameters and their meaning}

This simulation framework has five hyperparameters. The first one is $\alpha$ which is the stepsize of Algorithm 1 and dictates the change at every iteration. A small choice of $\alpha$ typically corresponds to slower convergence while a larger $\alpha$ allows larger changes at each iteration. Nevertheless, the value of $\alpha$ is always bounded from above since very large steps result in the instability of the simulation. This corresponds well with the stepsize in implicit numerical integration methods.

The second hyperparameter is the sampling frequency $F_s$ and is used to dictate the resolution of the solution signal. A large $F_s$ corresponds to detailed sampling and consequently enlarges the finite dimension of all the operators that have to act on the solution and auxiliary vectors. On the other hand, a small $F_s$ speeds up the computational cost of every iteration but, it might provide rough estimations of the solution. A very small choice of $F_s$ might cause instability in the case of a bursting neuron due to clear computational errors.

The third hyperparameter is the maximum number of iteration counts allowed. This serves as an upper limit for the number of iterations before forceful termination. Depending on the complexity of the single neuron and the quality of the initial condition, the authors recommend 500 to 7500 as a suitable number. Bear in mind that the termination of Algorithm 1 should ideally come from the criterion of subsection \ref{criterion}.

The hyperparameter $\lambda$ is used to shift the operators $\operatorname{F}_i$ and $\operatorname{G}_i$ into monotonicity. Therefore, its lower bound must always be chosen in a way that guarantees monotonicity. Study \cite{chaffey2023monotone} offers this limit for the operators of this paper. Furthermore, as discussed in (\ref{limitsofshift}), the right choice of $\lambda$ allows the choice of a larger $\alpha$ that accelerates the convergence of the algorithm. Nevertheless, a very large $\lambda$ fades the effect of the original operator and thus, more iterations are needed for convergence.

The final hyperparameter is the {\textit{simulation time}}. Similar to $F_s$, the simulation time also affects the length of the solution vector and a very large simulation time slows down the simulation significantly. Nevertheless, since this paper targets events, the simulation always starts at rest and ends at rest. Thus, the simulation time must be long enough to capture the full event. Furthermore, due to the use of FFT and iFFT, there's an assumption of periodicity (which is aligned with the rest-to-rest framework). This means that if the simulation time is not long enough, the solution cannot go from rest to rest but the solver enforces this and thus, the solution would be inaccurate at best.

\subsubsection{Termination Criterion}{\label{criterion}}

The termination criterion for Algorithm 1 is similar to most iterative algorithms and relies on the relative change of the solution at every iteration. Here, the $\frac{\norm{x^{k+1} - x^{k}}}{\norm{x^{k}}}$ determines the change of the solution after one iteration. If this change is less than a predetermined tolerance $\epsilon$, then the simulation can be terminated since the new iterations do not provide significant change. 

\begin{remark}
    The predetermined tolerance $\epsilon$ for the termination of the iterative Algorithm 1 relies heavily on $\alpha$. This stems from the fact that the stepsize $\alpha$ controls the size of the variations of the vectors $x^k$ and $z^k$ in Algorithm 1. 
\end{remark}

\subsubsection{Verification of the Solution}

It is important to verify if the solution $x^{sol}$ of Algorithm 1 is close enough to the true solution. In doing so, it is possible to feed this solution to the initial ZFP of (\ref{eq:method}). In an ideal scenario, the solution must be the zero vector of appropriate size. However due to numerical errors, certain deviations are expected. Thus, a good metric would be to check the norm of the term $\operatorname{E}(x^{sol}) + \sum\limits_{i=1}^{n}(\operatorname{F}_i(x^{sol})-\operatorname{G}_i(x^{sol}))$, where $x^{sol}$ denotes the solution vector of the iterative solver. If this norm is small, then the solution is of decent quality. In our work, we use the normalized norm $\frac{1}{L}(\operatorname{E}(x^{sol}) + \sum\limits_{i=1}^{n}(\operatorname{F}_i(x^{sol})-\operatorname{G}_i(x^{sol})))$ where $L$ is the length of $x^{sol}$ vector. 

Although this metric provides a distance to the true solution, it is not used for the purpose of termination. The justification for not using this metric for termination is purely computational. Evaluating the true solution norm at every iteration imposes a relatively expensive computational burden to every iteration of our algorithm while the alternative criterion does not.

\subsection{Illustrative Examples:}

\subsubsection{Single Spiking Neuron}

A single spiking neuron requires two timescales (instantaneous and fast). The numerical values for a spiking neuron are directly extracted from \cite{ribar2021neuromorphic} and the model is
\begin{equation}{\label{exam:spike}}
    C\operatorname{D}v + v - 2 \tanh(v) + 2 \tanh (\frac{1}{50 \operatorname{D} + \operatorname{Id}}v) - i_{ext} = 0
\end{equation}
where again $\operatorname{D}$ is the differentiation operator. With a reasonable shift of $\lambda$ we can rewrite the model as
\begin{equation}
    \begin{aligned}
        C\operatorname{D}v + (v - i_{ext}) - 2 \tanh(v) + \\ 
        (2 \tanh (\frac{1}{50 \operatorname{D} + \operatorname{Id}}v) + \lambda \operatorname{Id}) - (\lambda \operatorname{Id})  = 0
    \end{aligned}
\end{equation}
where this corresponds to $\operatorname{E} = C \operatorname{D}(\cdot)$, $\operatorname{F}_1 = \operatorname{Id} - i_{ext}$, $\operatorname{G}_1 = 2 \tanh(\cdot)$, $\operatorname{F}_2 = 2 \tanh (\frac{1}{50 \operatorname{D} + \operatorname{Id}}(\cdot)) + \lambda \operatorname{Id}$ and $\operatorname{G}_2 = \lambda \operatorname{Id}$ according to the language of (\ref{eq:method}) and Algorithm 1. Note that all the operators are monotone. $i_{ext}$ in $\operatorname{F}_1$ is nothing but an offset, a predetermined external input applied to the system to excite it.

The purpose of this example is twofold. First, it illustrates the capability of the proposed method to accurately capture the {\emph{excitability}} of spiking neurons. It delineates that small external currents only cause small perturbations (subthreshold behavior) whereas slightly larger currents result in large transient events (suprathreshold behavior), known as {\textit{spikes}}. This is the manifestation of excitability.  Figure \ref{fig:Spiking} demonstrates the excitability and rebound spiking of a simple neuron. It also shows the evolution of the solution vector at different iterations and its convergence toward the true solution, verified by numerical integration methods.

All the simulations begin at the resting potential of the neuron and also finish there. This temporal boundary condition attempts to restrict our framework to \emph{event} simulations.

\begin{figure}[h]
    \centering
    \includegraphics[width=0.4\textwidth]{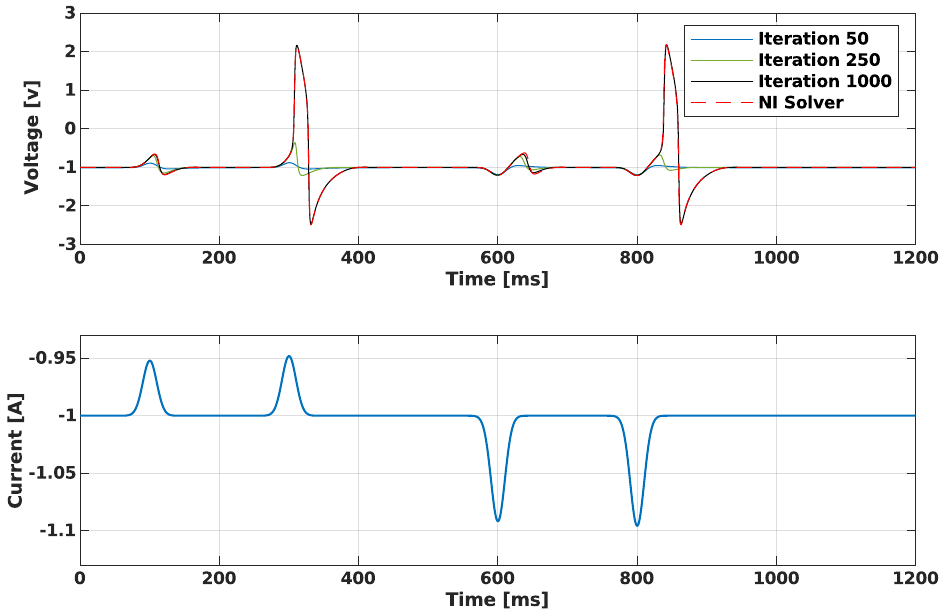} 
    \caption{Simulation of a spiking neuron subjected to external current. The external injected current cannot excite the system if it is small and only creates small bumps in the potential. In contrast, by having a slightly larger external current, the neuron spikes. As it can be seen, the signal gets closer to the solution as the algorithm progresses and, it converges to it. These results are verified by an NI solver.}
    \label{fig:Spiking}
\end{figure}

The second purpose of this example is to demonstrate the efficiency and simplicity of {\textit{variability analysis}} through continuation methods in this framework. The operator-theoretic representation of (\ref{EQ:compact}) provides a signal-to-signal mapping of inputs to outputs. Thus, a continuous change in the inputs or model parameters manifests itself as a continuous change of the output signal. This continuity is key for neuromodulation studies \cite{lee2012neuromodulation}.

Figure \ref{fig:Var} illustrates continuation on the modulation of the maximal  conductance parameter alpha from $\operatorname{G}_1 = 2 \tanh(\cdot)$ to $\operatorname{G}_1 = 2.4 \tanh(\cdot)$ and the question is how this changes the behavior. The algorithm is first implemented to simulate the behavior of the spiking neuron with $\operatorname{G}_1 = 2 \tanh(\cdot)$. After convergence, its coefficient is changed by 0.01 increments, and at each run, the solution of the previous simulation is fed to the solver as the initial condition. It is observed that the solver converges quickly to the next solution due to this continuity.

\begin{figure}[h]
    \centering
    \includegraphics[width=0.4\textwidth]{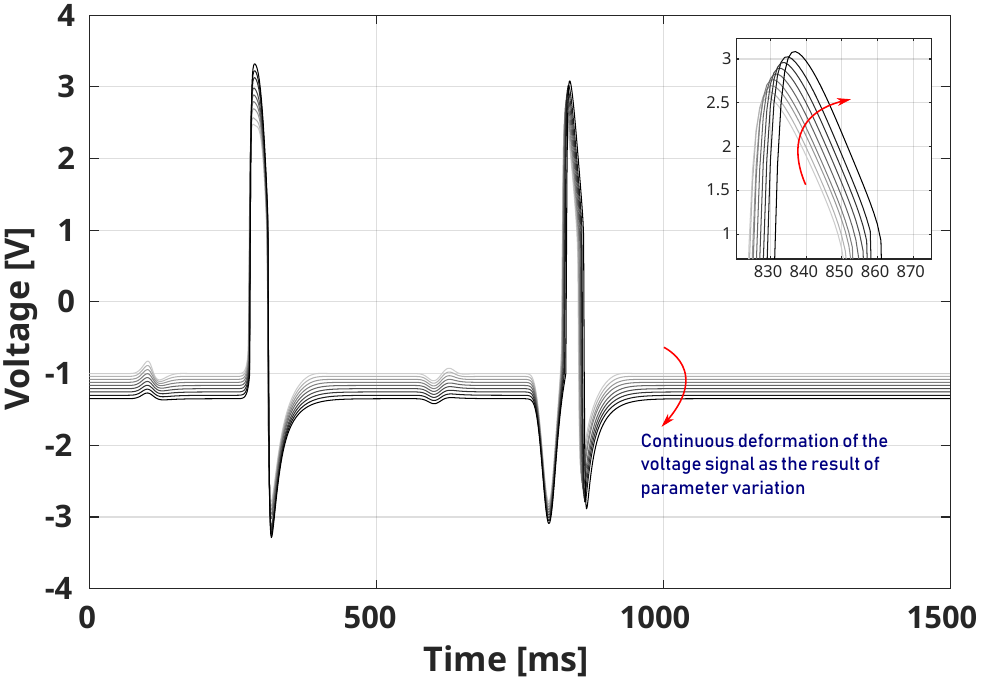} 
    \caption{Variation of the behavior of the spiking neuron as a result of an alteration of parameters. The parameter $\alpha ^- _f$ has been changed from $-2$ to $-2.4$ where the lightest trajectory corresponds to $-2$ and the darkest represents $-2.4$. The magnified subfigure demonstrates the rise of the amplitude as the result of this variation. The arrows indicate the direction of solutions' evolutions with parameter variation.}
    \label{fig:Var}
\end{figure}

The hyperparameters of Algorithm 1 for this example are chosen within the rational bounds. Nevertheless, they are not tuned for optimal performance. Here, $\lambda = 4$ to ensure monotonicity. $\alpha = 0.5$ to have a decent stepsize for speedy convergence. The sampling frequency is set to be $F_s = 10$ for the simulation of Fig \ref{fig:Spiking} but is reduced to $F_s = 1$ to show the robustness of the algorithm with coarser resolutions in the case of the variability analysis example. The simulation time is 1200 milliseconds in the spiking simulation and 1500 milliseconds in the variability analysis run. The algorithm terminates after 1000 iterations. The interested reader is referred to the literature of convex optimization for insight, bounds and tuning of these hyperparameters \cite{bertsekas2003convex,shen2016disciplined}.

\subsubsection{Single Bursting Neuron}

The model of the bursting neuron follows the same structure but requires three timescales (instantaneous, fast and slow). The model of interest is
\begin{equation}
    \begin{aligned}
        & C\operatorname{D}v + v - 2 \tanh(v) + 2 \tanh (\frac{1}{50 \operatorname{D} + \operatorname{Id}}v)  \\
        & - 1.5 \tanh (\frac{1}{50 \operatorname{D} + \operatorname{Id}}v + 0.88) + 1.5 \tanh (\frac{1}{2500 \operatorname{D} + \operatorname{Id}}v) \\
        & - i_{ext} = 0
    \end{aligned}
\end{equation}
where the parameters are again extracted from \cite{ribar2021neuromorphic}. It is straightforward to extract $\operatorname{E}$, $\operatorname{F}_i$ and $\operatorname{G}_i$ according to the previous example. Similar procedures to the previous example are taken and, the results are demonstrated in Fig \ref{fig:Bursting}. This figure illustrates perfect agreement between the results of the proposed simulation framework and the results of a NI solver that is presented for comparison. 

The purpose of this example is to demonstrate the capability of this solver in capturing multiscale behaviors, such as bursting. The estimated solution of the operator-theoretic algorithm naturally evolves from a coarse estimate of the event at iteration 300 to a detailed simulation of the burst at iteration 6400. In fact, from early on, the existence of an event at around $t = 6000 \ ms$ can be seen and, the later iterations only refine the details. 

\begin{figure}[h]
    \centering
    \includegraphics[width=0.4\textwidth]{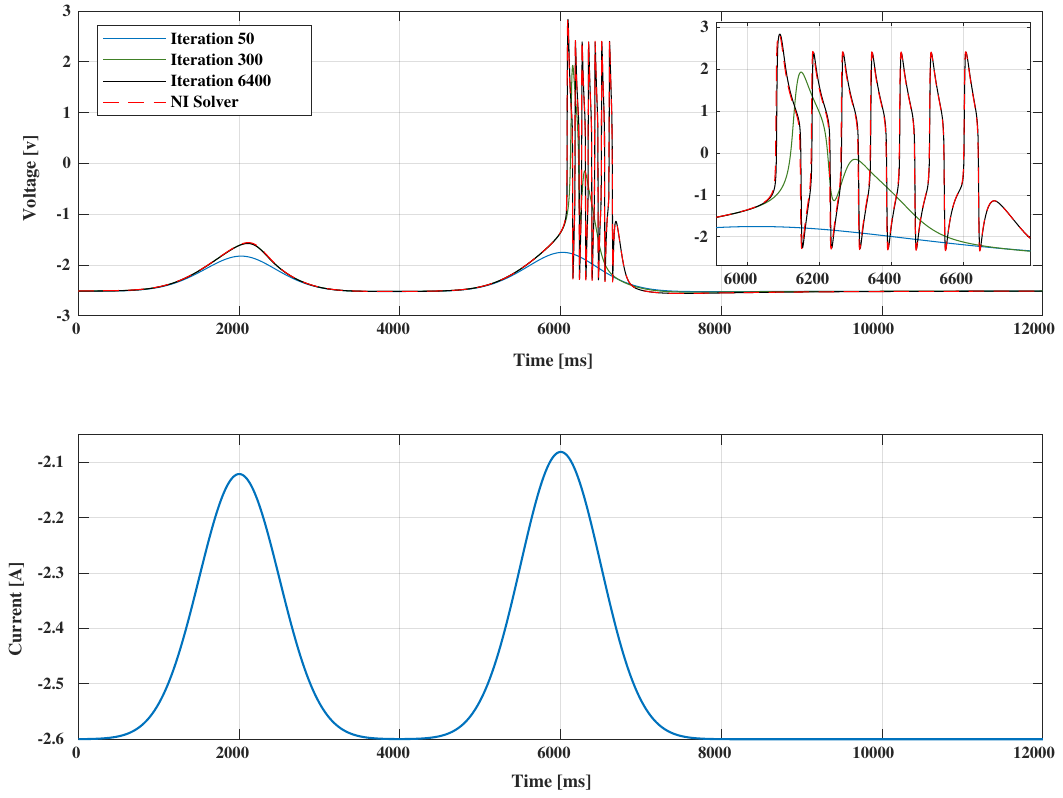} 
    \caption{The excitability of a bursting neuron is captured using this simulation framework. The simulation begins at rest and from early on, begins to detect the occurrence of events. By continuing the iterations, the solver refines the signal further until its convergence. NI solvers are used to verify the validity of the converged solution.}
    \label{fig:Bursting}
\end{figure}

Here, $\lambda = 2$ to ensure monotonicity. $\alpha = 0.15$ to have a decent stepsize for speedy convergence. The sampling frequency is set to be $F_s = 4$ and the simulation time is 12000 milliseconds. The algorithm terminates after 7000 iterations.

\subsubsection{Events in Half-center oscillators}

The final example of this paper is concerned with simulating a simple network. Bursting neurons are interconnected through inhibitory synaptic connections. The models for neurons are identical and demonstrated as
\begin{equation}
    \begin{aligned}
        & \operatorname{N}(v) = C\operatorname{D}v + v - 2 \tanh(v) + 2 \tanh (\frac{1}{50 \operatorname{D} + \operatorname{Id}}v)  \\
        & - 1.5 \tanh (\frac{1}{50 \operatorname{D} + \operatorname{Id}}v + 0.88) + \tanh (\frac{1}{2500 \operatorname{D} + \operatorname{Id}}v + 0.88)
    \end{aligned}
\end{equation}
and the synaptic connections have the format of (\ref{EQ:SynapseOT_Def}) as
\begin{align}\label{eq:SynapseOT}
    \operatorname{I}_{12} ^{syn}(\cdot) = 0.8 \operatorname{X}_{sig} (2 (v_2 - 1))
\end{align}
and
\begin{align}\label{eq:SynapseOT22}
    \operatorname{I}_{21} ^{syn}(\cdot) = 0.8 \operatorname{X}_{sig} (2 (v_1 - 1))
\end{align}
and the externally injected currents are consistent with the numerical values of Fig. (\ref{fig:HCO}) to start from equilibria.

A detailed scenario is devised to illustrate the rich behavior of the HCO network. The example begins by applying an external inhibitory current to the first neuron, driving it into a hyperpolarized state. Once the external current is removed, the first neuron undergoes a rebound burst. This burst inhibits the second neuron via its inhibitory synaptic connection, causing a rebound burst in the second neuron as well. Two different approaches are explored to simulate the HCO network and obtain this behavior.

Given the complexity of the network's behavior, solving it with a fine resolution from the outset is computationally inefficient. A fine resolution results in a larger solution vector and corresponding operators, significantly slowing each iteration. Moreover, this brute-force method disregards our emphasis on the modularity of event-based neuromorophic behaviors. Instead, starting with a coarse resolution provides an approximate solution (that captures the \emph{occurrence} of the events) and can be used as initialization of the algorithm with fine-resolution settings. This is illustrated in Fig. \ref{fig:HCO}.

\begin{figure}[h]
    \centering
    \includegraphics[width=0.45\textwidth]{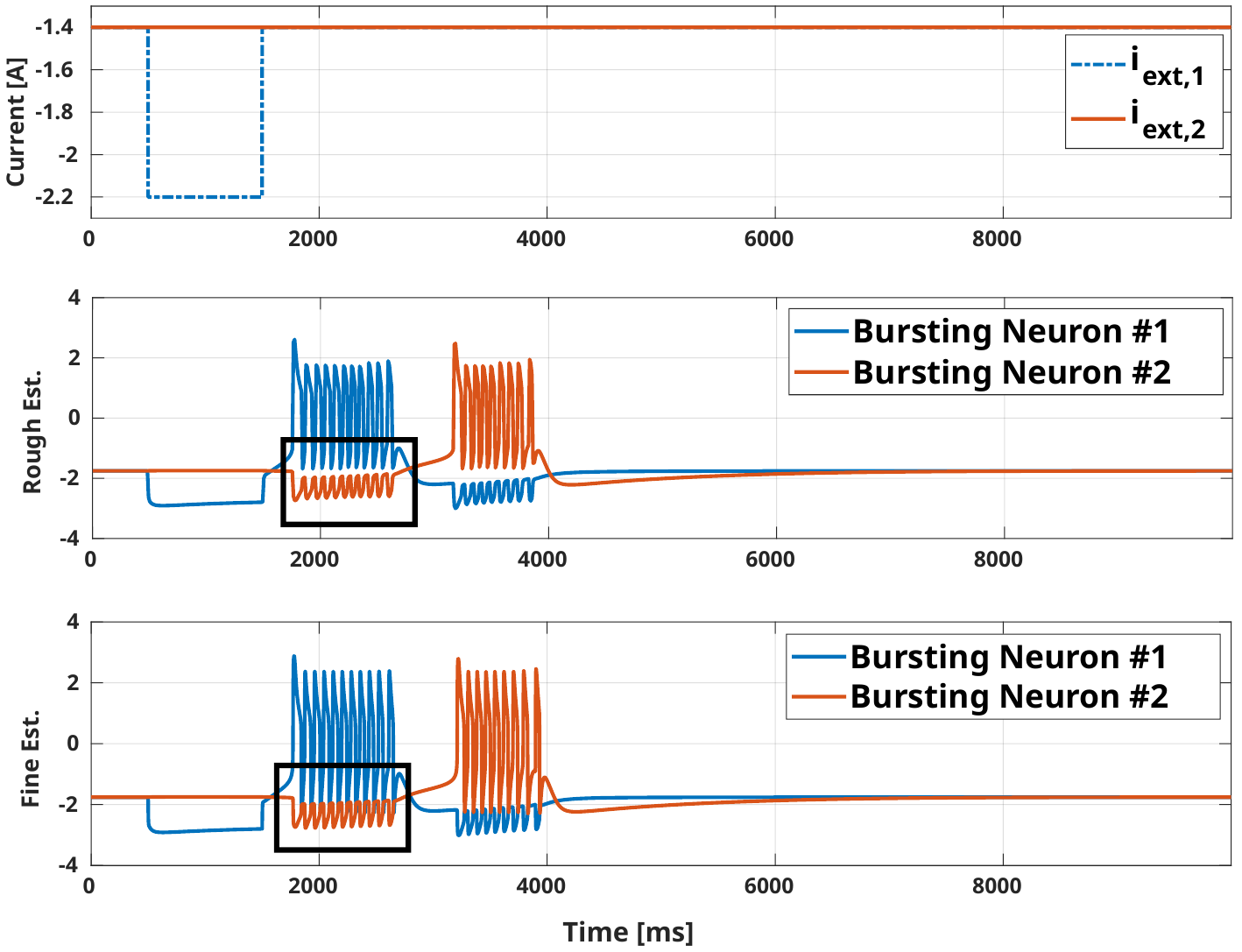} 
    \caption{The simulation of an HCO experiencing bursting events. The top figure illustrates the injected current to both neurons. The middle figure is the rough estimation of the solution with a small sampling frequency ($F_s = 0.1 Hz$). The figure at the bottom starts with the solution of the rough estimation as an initial guess and runs the algorithm with a larger sampling frequency to fine-tune the solution ($F_s = 2 Hz$). The change in the amplitude and finer details are among the most noticeable added details, highlighted by a black rectangle.}
    \label{fig:HCO}
\end{figure}

The result of the simulation at coarse-resolution provides an interesting insight. Compared to the true solution, the coarse resolution simulation does not get the amplitude of the bursting neurons right. Nevertheless, unlike NI solvers, errors in the estimation of solutions do not propagate forward in time. This stems from the computation of the entire solution signal at one step (within Algorithm 1) and the signal-to-signal mapping nature of this operator-theoretic approach. Consequently, and despite the first burst not being precisely estimated, the rebound burst of the second neuron is captured. This demonstrates that efficient coarse resolution simulations are effective in estimating the timing and location of key events that are the most significant characteristic of event-based systems, emphasizing further on the event-focused nature of the solver discussed in the bursting example.

An alternative approach to obtaining a refined solution of spiking networks leverages the modular nature of the circuit: the network couples the events of the individual neurons, enabling a warm-start of the network simulation. This pre-existing knowledge is information on the shape and form of the individual event (spiking or bursting) that can serve as the template. To this end, the simulation is first performed on a coarse resolution and the results are used to identify the approximate locations of events. Then, the coarse resolution events are substituted with precise templates (derived from the results of single neuron simulations). This approach enforces a more accurate representation of individual events within the system.

The hyperparameters for the rough and fine simulation can be found in the code appended to this paper.

\begin{remark}
    Given the sharp nature of spikes, it is very difficult to capture their exact shape in simulation. Thus, to get the shapes of spikes exactly right, a large sampling frequency is needed. However, it is important to note that in most cases, the shape of the spike bears no significant information, and it is only the occurrence of the event that bears meaning. NI methods cannot circumvent this issue due to forward error propagations and thus, face computational bottlenecks. Nevertheless, the event-based nature of this solver evades this issue as depicted in the HCO example.
\end{remark}

As the network size increases, the number of synaptic connections grows quadratically, which is a primary factor limiting scalability in most numerical algorithms for simulating large-scale spiking networks. In the presented framework, however, the number of smaller operators grows linearly with the network size. While this mitigates some challenges, it enlarges the consensus set $ \boldsymbol{C}$ (that our distributed algorithm relies upon), slowing convergence as the network expands. Although this permits the use of highly parallel implementation, the proposed algorithm does not resolve the problem of scalability.

\section{Conclusion}

This paper proposes a novel framework to simulate neuromorphic systems. This framework relies on an operator-theoretic representation of the neuromorphic circuit, intertwined with an iterative solver. This is in contrast to the prevalent incremental numerical integration methods of the literature that do not exploit the event-based nature of the neuronal behavior (limited to spiking and bursting), do not offer any methodological approach for variability analysis, and do not permit the incorporation of prior knowledge of the system into the solver. These shortcomings were the main motivation for the operator-theoretic representation of this paper.

In the proposed framework, the operator-theoretic representation of the neuron is split into several smaller operators based on the topology of the network and its single-layered architecture. These smaller operators are used within the structure of a new difference-of-monotone solver and the solution trajectory of the network is obtained. It is shown that this method is event-based, and capable of using prior knowledge of the spiking system in its structure. Variability analysis is also efficient and straightforward via this method.

Future research will refine the proposed methodology to make it scalable to large networks. A crucial element to that end will be to revisit the model of the single elements (channels and synapses) so as to enforce a monotonicity property for each element. A step in that direction is provided in \cite{shahhosseini2025variable}.

\section{Code Availability}

The MATLAB code for every example of the simulation section can be found \href{https://github.com/AmirShahhosseini/Splitting_IEEE_TAC}{here}. The repository is purposefully developed to be accessible, readable, and reusable.

\appendices

\section{Proof of Convergence}\label{appendix}

The proof of convergence for Theorem \ref{thm:main} is presented through two steps. The first step proves the {\textit{local}} convergence of the three operator splitting algorithm on the subset ${\mathcal{D}}$ for a difference of monotone setup. The second part proves the convergence for the lifted problem for arbitrary number of operators.

\subsection{Problem Statement and Deriving the Algorithm}

Although the derivation of this 3-operator splitting might resemble the DYS setup, as shown in \cite{ryu2022large}, it differs in both its functionality and proof of convergence. The DYS algorithm was originally proposed for a fully monotone setup \cite{davis2017three} with 3 monotone smaller operators (where some smaller operators bear extra properties as well). The monotonicity of all smaller operators and the extra properties are principal in its proof of convergence. On the other hand, the proposed algorithm of this paper tackles the problem of difference of monotone ZFPs. This is an entirely different setup and affects the proof of convergence. Nevertheless, the assumption here is local monotonicity that conforms to the physics of our problem and is nearly always the case with difference of convex algorithms \cite{yuille2003concave}.

The problem of interest is to find a fixed-point iteration (FPI) algorithm that obtains the zero of the 

\begin{equation}{\label{Basic_ZFP}}
    (\operatorname{A} + \operatorname{B} - \operatorname{C})x = 0
\end{equation}
where $x$ is a signal in the Hilbert space of $\mathcal{L}^2_{[0,\infty)}$. The properties of these operators are discussed in detail within Theorem (\ref{thm:three}).

Reformulating the problem to 
\begin{equation}
    (\operatorname{Id}+\alpha \operatorname{A}) x-(\operatorname{Id}-\alpha \operatorname{B}) x-\alpha \operatorname{C} x = 0
\end{equation}
and using the definition of the reflected resolvent (Definition \ref{def:reflected}), it is possible to rewrite the expression as
\begin{equation}
    (\operatorname{Id}+\alpha \operatorname{A}) x-\operatorname{R}_{\alpha \operatorname{B}}(\operatorname{Id}+\alpha \operatorname{B}) x-\alpha \operatorname{C} x = 0
\end{equation}

By introducing the auxiliary variable $z=(\operatorname{Id}+\alpha \operatorname{B}) x $, it is possible to obtain
\begin{equation}
    (\operatorname{Id}+\alpha \operatorname{A}) x-\operatorname{R}_{\alpha \operatorname{B}} z-\alpha \operatorname{C} x = 0
\end{equation}
and reformulate this ZFP in terms of the auxiliaryvariable as
\begin{equation}
    \left(\operatorname{R}_{\alpha \operatorname{B}}+\alpha \operatorname{C J}_{\alpha \operatorname{B}}\right) z =(\operatorname{Id}+\alpha \operatorname{A}) \operatorname{J}_{\alpha \operatorname{B}} z
\end{equation}
and this can be manipulated into
\begin{equation}
    \operatorname{J}_{\alpha \operatorname{A}}\left(\operatorname{R}_{\alpha \operatorname{B}}+\alpha \operatorname{C} \operatorname{J}_{\alpha \operatorname{B}}\right) z=\operatorname{J}_{\alpha \operatorname{B}} z
\end{equation}
or more compactly,
\begin{equation}
    \left(\operatorname{R}_{\alpha \operatorname{A}}\left(\operatorname{R}_{\alpha \operatorname{B}}+\alpha \operatorname{C} \operatorname{J}_{\alpha \operatorname{B}}\right)+\alpha \operatorname{C} \operatorname{J}_{\alpha \operatorname{B}}\right) z=z
\end{equation}
and to write in an averaged formalism, it is possible to rewrite it as
\begin{equation}
    (\frac{1}{2} \operatorname{Id}+ \frac{1}{2} \operatorname{Q}) z=z
\end{equation}
where
\begin{equation}
    \operatorname{Q}=\operatorname{R}_{\alpha \operatorname{A}}\left(\operatorname{R}_{\alpha \operatorname{B}}+\alpha \operatorname{C} \operatorname{J}_{\alpha \operatorname{B}}\right)+\alpha \operatorname{C J}_{\alpha \operatorname{B}}
\end{equation}

Operator $\operatorname{T}$, defining the fixed-point iteration, is now defined using $\operatorname{Q}$ and can be expanded as 
\begin{equation}{\label{eq:def_averaged}}
    \operatorname{T} = \frac{1}{2} \operatorname{Id}+\frac{1}{2} \operatorname{Q}= \operatorname{Id}-\operatorname{J}_{\alpha \operatorname{B}}+\operatorname{J}_{\alpha \operatorname{A}}\left(2\operatorname{J}_{\alpha \operatorname{B}} - \operatorname{Id} + \alpha \operatorname{C} \operatorname{J}_{\alpha \operatorname{B}} \right)
\end{equation}
to ensure the averagedness of the fixed-point iteration's mapping. It progresses as
\begin{equation}{\label{EQ:Basic_FPI}}
    z_{k+1} = \operatorname{T}z_k
\end{equation}
where this can be written as an algorithm, purely based on the resolvent operator, as
\begin{equation}{\label{eq:algorithm}}
\begin{aligned}
& x_{k+1 / 2}=\operatorname{J}_{\alpha \operatorname{B}}\left(z_{k}\right) \\
& x_{k+1}=\operatorname{J}_{\alpha \operatorname{A}}\left(2 x_{k+1 / 2}-z_{k}+\alpha \operatorname{C} x_{k+1 / 2}\right) \\
& z_{k+1}=z_{k}+x_{k+1}-x_{k+1 / 2}
\end{aligned}
\end{equation}
which is similar to the difference of monotone Douglas-Rachford algorithm of \cite{chuang2022unified}, derived and studied for when the operators are the gradients of convex functions, imposing the extra condition of cyclic monotonicity.

\subsection{Proof of Convergence for Three Operators}

We begin by recalling the Krasnosel'skiĭ-Mann Theorem.
\begin{proposition}{\label{prop1}}
    Assume a $\lambda$-averaged operator $\operatorname{T}: \mathcal{D} \rightarrow \mathcal{D}$, where $\mathcal{D}$ is closed, convex and non-empty subset of $\mathcal{H}$ and $\text{Fix} \operatorname{T} \neq \varnothing$. The fixed-point iteration of
\begin{equation}{\label{eq:proposition}}
    x_{n+1} = \operatorname{T} x_n
\end{equation}
generates a sequence $(x_n)_{_{n \in \mathbb{N}}}$ that
converges weakly to $x^* \in \text{Fix} \operatorname{T}$.
\end{proposition}
\begin{proof}
    Set $N = (1/\lambda)(\operatorname{T} + (\lambda - 1)\operatorname{Id})$ in \cite[Thm. 5.14]{bauschke_convex_2011}.
\end{proof}

\begin{theorem}{\label{thm:three}}
    Let ${\mathcal{D}}$ be a closed, convex and non-empty subset of $\mathcal{H}$, $\alpha > 0$ and let $\operatorname{A}$, $\operatorname{B}$ and $\operatorname{C} : \mathcal{H} \xrightarrow{} \mathcal{H}$ satisfy the assumptions:

    \begin{enumerate}
        \item Operator $\operatorname{A}$ is $\rho$-strongly monotone on $\mathcal{D}$
        \item Operator  $\operatorname{B}$ is $\gamma$-strongly monotone on $\mathcal{D}$
        \item Operator $\operatorname{C}$ satisfies
        \begin{equation}
             (\beta + \frac{1}{\epsilon})\norm{y - \hat{y}}^2 \geq \ip{\operatorname{C}y - \operatorname{C}\hat{y}}{y-\hat{y}} \geq \beta \norm{y - \hat{y}}^2
        \end{equation}
        for all $y, \hat y \in \mathcal{D}$.
        \item Operator $\operatorname{T}$, defined by (\ref{eq:def_averaged}), satisfies $ \operatorname{T}(\mathcal{D}) \subseteq \mathcal{D}$
        \item $\gamma > \beta + 1 + \frac{2}{\epsilon}$ 
        \item $\text{Fix} \operatorname{T} \neq \emptyset$
    \end{enumerate}
    then, if $0 \leq \alpha \leq \frac{2}{\beta^2}$, the operator $\operatorname{T}$ is $\theta$-averaged with $\theta = \frac{2}{4 - \alpha\beta^2}$ and the FPI \eqref{eq:proposition} converges weakly to $x^* \in \operatorname{Fix} T$.
\end{theorem}

{\textit{Proof:}} By proposition \ref{prop1}, this theorem holds if the operator $\operatorname{T}$ is averaged.

Given (\ref{eq:algorithm}), the first iteration of the algorithm is
\begin{equation}{\label{eq:algorithm1}}
\begin{aligned}
& x_{1 / 2}=\operatorname{J}_{\alpha \operatorname{B}}\left(z_{0}\right) \\
& x_{1}=\operatorname{J}_{\alpha \operatorname{A}}\left(2 x_{1 / 2}-z_{0}+\alpha \operatorname{C} x_{1 / 2}\right) \\
& z_{1}=z_{0}+x_{1}-x_{1 / 2}
\end{aligned}
\end{equation}

From \cite[Lem. 5, Ch. 13]{ryu2022large}, averagedness of $\operatorname{T}$ is equivalent to
\begin{equation}
    \small
    \norm{\operatorname{T}z_0 - \operatorname{T}\hat{z}_0}^2 \leq \norm{z_0 - \hat{z}_0}^2 - \frac{1-\theta}{\theta} \norm{\operatorname{T}z_0 - z_0 - \operatorname{T}\hat{z}_0 + \hat{z}_0}^2
\end{equation}
where this can be rewritten using (\ref{EQ:Basic_FPI}) to 
\begin{equation}{\label{eq:inequality}}
    \norm{z_1 - \hat{z}_1}^2 \leq \norm{z_0 - \hat{z}_0}^2 - \frac{1-\theta}{\theta} \norm{z_1 - z_0 - \hat{z}_1 + \hat{z}_0}^2.
\end{equation}

The proof now consists of showing the inequality of (\ref{eq:inequality}) holds given the assumptions. It is possible to rewrite the LHS of (\ref{eq:inequality}) as
\begin{equation}
    \begin{aligned}
    & \norm{z_1 - \hat{z}_1}^2 = \norm{{z}_0 + ({z}_1 - {z}_0) - (\hat{z}_0 + (\hat{z}_1 - \hat{z}_0))}^2  \\ 
    & =\norm{z_0 - \hat{z}_0}^2  +\norm{(z_1 - z_0) - (\hat{z}_1 - \hat{z}_0)}^2 \\
    & + 2 \ip{z_0 - \hat{z}_0}{(z_1 - z_0) - (\hat{z}_1 - \hat{z}_0)} 
    \\
    & = \norm{z_0 - \hat{z}_0}^2  -\frac{1-\theta}{\theta}\norm{(z_1 - z_0) - (\hat{z}_1 - \hat{z}_0)}^2 + \\ 
    & \frac{1}{\theta}\norm{(z_1 - z_0) - (\hat{z}_1 - \hat{z}_0)}^2 \\
    & + 2 \ip{z_0 - \hat{z}_0}{(z_1 - z_0) - (\hat{z}_1 - \hat{z}_0)}
    \end{aligned}
\end{equation}

and now, all that needs to be shown is 

\begin{equation}{\label{eq:main}}
    \begin{aligned}
    & \frac{1}{\theta}\norm{(z_1 - z_0) - (\hat{z}_1 - \hat{z}_0)}^2 \\ 
    & + 2 \ip{z_0 - \hat{z}_0}{(z_1 - z_0) - (\hat{z}_1 - \hat{z}_0)} \leq 0
    \end{aligned}
\end{equation}

for (\ref{eq:inequality}) to hold. We set $\theta = \frac{2}{4 - \alpha \beta^2}$ and note that the assumption $0 \leq \alpha \leq \frac{2}{\beta^2}$ guarantees $0 \leq \theta \leq 1$. We now rewrite (\ref{eq:main}) as follows and show that for this choice of $\theta$, the inequality 
\begin{equation}
    \begin{aligned}
    &(2 - \frac{\alpha \beta^2}{2})\norm{(z_1 - z_0) - (\hat{z}_1 - \hat{z}_0)}^2 \\
    & + 2 \ip{z_0 - \hat{z}_0}{(z_1 - z_0) - (\hat{z}_1 - \hat{z}_0)} \leq 0,
    \end{aligned}
\end{equation}
holds. Using the last relation of (\ref{eq:algorithm1}), this can be rewritten as
\begin{equation}
    \begin{aligned}
    & (2 - \frac{\alpha \beta^2}{2})\norm{(x_1 - x_{1 / 2}) - (\hat{x}_1 - \hat{x}_{1/2})}^2 + \\
    & 2 \ip{z_0 - \hat{z}_0}{(x_1 - x_{1 / 2}) - (\hat{x}_1 - \hat{x}_{1/2})} \leq 0
    \end{aligned}
\end{equation}
where this can be rewritten as
\begin{equation}{\label{eq:long}}
\begin{aligned}
    & (2 - \frac{\alpha \beta^2}{2}) \norm{(x_1 - x_{1 / 2}) - (\hat{x}_1 - \hat{x}_{1/2})}^2 + \\
    & 2 \ip{z_0 - \hat{z}_0 - \alpha (\operatorname{C}{x}_{1/2} - \operatorname{C}\hat{x}_{1/2})}{(x_1 - x_{1 / 2}) - (\hat{x}_1 - \hat{x}_{1/2})} \\
    & + 2 \alpha \ip{\operatorname{C}{x}_{1/2} - \operatorname{C}\hat{x}_{1/2}}{(x_1 - x_{1 / 2}) - (\hat{x}_1 - \hat{x}_{1/2})}  + \\
    & \frac{2 \alpha}{\beta^2} \norm{\operatorname{C}{x}_{1/2} - \operatorname{C}\hat{x}_{1/2}}^2 - \frac{2 \alpha}{\beta^2} \norm{\operatorname{C}{x}_{1/2} - \operatorname{C}\hat{x}_{1/2}}^2 \leq 0
    \end{aligned}
\end{equation}

Applying the Cauchy-Schwarz inequality and $\beta$-strong monotonicity of $\operatorname{C}$, gives
\begin{equation}
    \norm{\operatorname{C}{x}_{1/2} - \operatorname{C}\hat{x}_{1/2}}  \geq \beta \norm{{x}_{1/2} -\hat{x}_{1/2}}
\end{equation}
From this, $2 \alpha \norm{{x}_{1/2} -\hat{x}_{1/2}}^2 - \frac{2 \alpha}{\beta ^2}  \norm{\operatorname{C}{x}_{1/2} - \operatorname{C}\hat{x}_{1/2}}^2 \leq 0$
and thus (\ref{eq:long}) can be rewritten as
\begin{equation}{\label{eq:proof}}
\begin{aligned}
    & (2 - \frac{\alpha \beta^2}{2}) \norm{(x_1 - x_{1 / 2}) - (\hat{x}_1 - \hat{x}_{1/2})}^2 \\
    & + 2 \ip{z_0 - \hat{z}_0 - \alpha (\operatorname{C}{x}_{1/2} - \operatorname{C}\hat{x}_{1/2})}{(x_1 - x_{1 / 2}) - (\hat{x}_1 - \hat{x}_{1/2})} \\
    & + 2 \alpha \ip{\operatorname{C}{x}_{1/2} - \operatorname{C}\hat{x}_{1/2}}{(x_1 - x_{1 / 2}) - (\hat{x}_1 - \hat{x}_{1/2})}  \\
    & + 2 \alpha \norm{{x}_{1/2} -\hat{x}_{1/2}}^2 - \frac{2 \alpha}{\beta^2} \norm{\operatorname{C}{x}_{1/2} - \operatorname{C}\hat{x}_{1/2}}^2 \leq 0
    \end{aligned}
\end{equation}

It is now possible to break (\ref{eq:proof}) into two parts and prove separately that each part is negative. It will be proven that

\begin{equation}{\label{eq:p1}}
\begin{aligned}
    &P_1 = 2 \norm{(x_1 - x_{1 / 2}) - (\hat{x}_1 - \hat{x}_{1/2})}^2
    +  2 \alpha \norm{{x}_{1/2} -\hat{x}_{1/2}}^2 + \\
    &2 \ip{z_0 - \hat{z}_0 - \alpha (\operatorname{C}{x}_{1/2} - \operatorname{C}\hat{x}_{1/2})}{(x_1 - x_{1 / 2}) - (\hat{x}_1 - \hat{x}_{1/2})}
    \end{aligned}
\end{equation}
and 
\begin{equation}{\label{eq:p2}}
\begin{aligned}
    &P_2 = - \frac{\alpha \beta^2}{2} \norm{(x_1 - x_{1 / 2}) - (\hat{x}_1 - \hat{x}_{1/2})}^2 \\
    & + 2 \alpha \ip{\operatorname{C}{x}_{1/2} -  \operatorname{C}\hat{x}_{1/2}}{(x_1 - x_{1 / 2}) - (\hat{x}_1 - \hat{x}_{1/2})} \\
    & - \frac{2 \alpha}{\beta^2} \norm{\operatorname{C}{x}_{1/2} - \operatorname{C}\hat{x}_{1/2}}^2
    \end{aligned}
\end{equation}
are nonpositive.

First, we begin with $P_2$ and prove that it is indeed always negative. We begin with factorizing by $-\frac{\alpha}{2}$ and 
\begin{equation}
\begin{aligned}
    &P_2 = -\frac{\alpha}{2}[ \beta^2 \norm{(x_1 - x_{1 / 2}) - (\hat{x}_1 - \hat{x}_{1/2})}^2 \\
    & -4 \ip{\operatorname{C}{x}_{1/2} -  \operatorname{C}\hat{x}_{1/2}}{(x_1 - x_{1 / 2}) - (\hat{x}_1 - \hat{x}_{1/2})} \\
    & + \frac{4}{\beta^2} \norm{\operatorname{C}{x}_{1/2} - \operatorname{C}\hat{x}_{1/2}}^2  ] \leq 0
    \end{aligned}
\end{equation}
where this can be reformulated to

\begin{equation}
    \begin{aligned}
    &V_1 = \beta ((x_{1 / 2} - \hat{x}_{1/2}) - ({x}_1 - \hat{x}_1)), \\
    &V_2 = \frac{2}{\beta}(\operatorname{C}{x}_{1/2} - \operatorname{C}\hat{x}_{1/2}), \\
    &P_2 = -\frac{\alpha}{2} ||V_1 + V_2 ||^2 \leq 0
    \end{aligned}
\end{equation}

The next step is to prove that the terms of (\ref{eq:p1}) are also negative. 
Factorizing by $-2$ gives
\begin{equation}{\label{eq:midexpanded}}
\small
\begin{aligned}
    &P_1 = -2 (-\norm{(x_1 - x_{1 / 2}) - (\hat{x}_1 - \hat{x}_{1/2})}^2 
    - \alpha \norm{{x}_{1/2} -\hat{x}_{1/2}}^2 \\
    &- \ip{z_0 - \hat{z}_0 - \alpha (\operatorname{C}{x}_{1/2} - \operatorname{C}\hat{x}_{1/2})}{(x_1 - x_{1 / 2}) - (\hat{x}_1 - \hat{x}_{1/2})})
\end{aligned}
\end{equation}
is nonpositive. Thus, 
\begin{equation}{\label{eq:expanded}}
    \begin{aligned}
    &P_1 = -2 (\ip{-(x_1 - \hat{x}_1)}{(x_1 - \hat{x}_1)} \\
    & + \ip{-(x_{1/2} - \hat{x}_{1/2})}{(x_{1/2} - \hat{x}_{1/2})} + \\
    & 2 \ip{x_1 - \hat{x}_1}{x_{1/2} - \hat{x}_{1/2}} - \alpha \norm{{x}_{1/2} -\hat{x}_{1/2}}^2 \\
    & - \ip{z_0-\hat{z}_0}{x_1 - \hat{x}_1} + \\
    & \ip{z_0-\hat{z}_0}{x_{1/2} - \hat{x}_{1/2}} + \ip{\alpha(\operatorname{C}{x}_{1/2} - \operatorname{C}\hat{x}_{1/2})}{x_1 - \hat{x}_1} - \\
    & \ip{\alpha(\operatorname{C}{x}_{1/2} - \operatorname{C}\hat{x}_{1/2})}{{x}_{1/2} -\hat{x}_{1/2}}) \leq 0
    \end{aligned}
\end{equation}

Bear in mind that the terms of (\ref{eq:algorithm}) can be rewritten as

\begin{equation}{\label{eq:temp}}
    \begin{aligned}
    &\operatorname{A}x_1 = \frac{1}{\alpha}(2x_{1/2} - z_0 + \alpha \operatorname{C}x_{1/2} - x_1) \\
    &\operatorname{B}x_{1/2} = \frac{1}{\alpha} (z_0 - x_{1/2})
    \end{aligned}
\end{equation}

Thus, (\ref{eq:expanded}) can be reformulated to resemble that of (\ref{eq:temp}) so the properties of the existing operators, on the subset $\mathcal{D}$ can be used. To this end, we can rewrite (\ref{eq:expanded}) as
\begin{equation}
    \begin{aligned}
    &-2 (\alpha \ip{\operatorname{A}x_1 - \operatorname{A}\hat{x}_1}{x_1 - \hat{x}_1} \\
    & + \alpha \ip{\operatorname{B}x_{1/2} - \operatorname{B}\hat{x}_{1/2}}{x_{1/2} - \hat{x}_{1/2}} \\
    &- \alpha (\ip{\operatorname{C}x_{1/2} - \operatorname{C}\hat{x}_{1/2}}{x_{1/2} - \hat{x}_{1/2}} + \norm{x_{1/2} - \hat{x}_{1/2}}^2)) \leq 0
    \end{aligned}
\end{equation}
where the first two terms are negative but the third term is not.  We then have
\begin{equation}
    \begin{aligned}
    & -2 \alpha (\ip{\operatorname{A}x_1 - \operatorname{A}\hat{x}_1}{x_1 - \hat{x}_1} \\
    & + \ip{\operatorname{B}x_{1/2} - \operatorname{B}\hat{x}_{1/2}}{x_{1/2} - \hat{x}_{1/2}} \\
    & +(\ip{\operatorname{C}x_{1/2} - \operatorname{C}\hat{x}_{1/2}}{x_{1/2} - \hat{x}_{1/2}} -\beta \norm{x_{1/2} - \hat{x}_{1/2}}^2)  \\
    & + (\beta - 1)\norm{x_{1/2} - \hat{x}_{1/2}}^2 \\
    & - 2  \ip{\operatorname{C}x_{1/2} - \operatorname{C}\hat{x}_{1/2}}{x_{1/2} - \hat{x}_{1/2}}) \leq 0
    \end{aligned}
\end{equation}
and use the strong monotonicity of operators $\operatorname{A}$, $\operatorname{B}$ and $\operatorname{C}$ on subset $\mathcal{D}$ for further insight. 
\begin{equation}
    \begin{aligned}
    & -2 \alpha(\rho \norm{x_1 - \hat{x}_1}^2 + \gamma \norm{{x_{1/2} - \hat{x}_{1/2}}}^2 \\
    & + (\beta - 1)\norm{x_{1/2} - \hat{x}_{1/2}}^2 \\
    & - 2  \ip{\operatorname{C}x_{1/2} - \operatorname{C}\hat{x}_{1/2}}{x_{1/2} - \hat{x}_{1/2}}) \leq 0
    \end{aligned}
\end{equation}
An upper bound for the last inner product is then given by
\begin{equation}
    \begin{aligned}{\label{eq:losslessness}}
    & -2 \alpha(\rho \norm{x_1 - \hat{x}_1}^2 + (\gamma - \beta -1 -\frac{2}{\epsilon} )\norm{x_{1/2} - \hat{x}_{1/2}}^2) \leq 0
    \end{aligned}
\end{equation}

And the final inequality holds due to the assumption 5 of the Theorem \ref{thm:three}.

\begin{remark}
    The monotonicity of the operators within the subset $\mathcal{D}$ might be seen as a strange or conservative requirement. Nevertheless, this is completely compatible with the physics of the problem that necessitates {\textit{local monotonicity}}. In this problem, spiking systems are physical systems, with a limited supply of internal energies. Thus, the energy landscape of the problem can be seen as a convex functional with bumps caused by the anti-monotone elements (with concave energy landscapes) that locally break this otherwise convex landscape.
\end{remark}

\begin{remark}
Strong monotonicity is a considerable assumption in the context of monotone zero finding problems (ZFP). However, this is not the case in difference of monotone ZFPs. This is because it is possible to add and subtract terms that make each operator strongly monotone. This is not possible in the case of monotone ZFPS since the addition of any strongly monotone term requires its subtraction (so the problem remains unchanged) and thus, an anti-monotone term is introduced that affects the monotone landscape of the aforementioned ZFP.
\end{remark}

\subsection{Proof of Convergence for the Lifted Algorithm}

The second step provides an extension from the three operator splitting algorithm to arbitrary number of operators, represented as
\begin{equation}{\label{eq:method_app}}
    \operatorname{E}(x) + \sum\limits_{i=1}^{M}(\operatorname{F}_i(x)-\operatorname{G}_i(x)) = 0
\end{equation}
The normal cone operator $\operatorname{N}_{\boldsymbol{C}}$ of the consensus set  $\boldsymbol{C}$ is defined by
\begin{equation}{\label{eq:normal_cone}}
    \begin{aligned}
        & \operatorname{N}_{\boldsymbol{C}} \boldsymbol{x}= \begin{cases}\left\{\boldsymbol{u} \in \boldsymbol{\mathcal{H}} \mid \sum_{i \in I} u_{i}=0\right\}, & \text { if } \boldsymbol{x} \in \boldsymbol{C} ; \\ \varnothing, & \text { otherwise. }\end{cases}
    \end{aligned}
\end{equation}

We now state a technical lemma.
\begin{lemma}{\label{lemma}}
    Given an operator $\boldsymbol{\operatorname{A}}$ in the lifted product space $\boldsymbol{\mathcal{H}} = \bigoplus_{i \in \{1,...,M\}} \mathcal{H}$ and defined as $\boldsymbol{\operatorname{A}} = [\operatorname{A}, \cdots, \operatorname{A}]^T$ where $\operatorname{A} : \mathcal{H} \xrightarrow{} \mathcal{H}$ is maximally monotone,  the identity 
    \begin{equation}
        \operatorname{J}_{\boldsymbol{\operatorname{A}} + \operatorname{N}_{\boldsymbol{C}}}(\boldsymbol{x}) = \operatorname{J}_{\boldsymbol{\operatorname{A}}}(\operatorname{P}_{\boldsymbol{C}}(\boldsymbol{x}))
    \end{equation}
    holds, where $\operatorname{P}_{\boldsymbol{C}}$ is defined in (\ref{eq:definitions}).
\end{lemma}

\begin{proof}
Writing
\begin{equation}
    \boldsymbol{y} = \operatorname{J}_{\boldsymbol{\operatorname{A}} + \operatorname{N}_{\boldsymbol{C}}}(\boldsymbol{x})
\end{equation}
we have
\begin{align}{\label{eq:lemma_eq1}}
     \boldsymbol{x} &= \boldsymbol{y} + \boldsymbol{\operatorname{A}}(\boldsymbol{y}) + \operatorname{N}_{\boldsymbol{C}}(\boldsymbol{y})\\
    \boldsymbol{y} &= \operatorname{J}_{\boldsymbol{\operatorname{A}}}(\boldsymbol{x} - \operatorname{N}_C(\boldsymbol{y}))\label{eq:lemma_eq2}
\end{align}
Given the definition of the normal cone (\ref{eq:normal_cone}), for this expression to be defined, we must have $\boldsymbol{y}\in\boldsymbol{{C}}$, so $\boldsymbol{y} = (y, \cdots, y)$ for some $y \in \mathcal{H}$. By \cite[part (vi) of Proposition 26.4]{bauschke_convex_2011}, we have $\operatorname{J}_{\boldsymbol{\operatorname{A}}}(\boldsymbol{x}) = \left( \operatorname{J}_{{\operatorname{A}_i}}({x}_i) \right)_{i = \{1,\cdots, M\}}$ and thus (\ref{eq:lemma_eq2}) can be reformulated as

\begin{equation}{\label{eq:lemma_eq3}}
    y_i = \operatorname{J}_{{\operatorname{A}}}(x_i - \operatorname{N}_C(\boldsymbol{y}))
\end{equation}
and since  $\boldsymbol{y}$ is in the consensus set, $y_i = y$ and this implies from (\ref{eq:normal_cone}) that $\operatorname{N}_C(\boldsymbol{y}) = \boldsymbol{u}$ for some $\boldsymbol{u}$, where by definition of the normal cone, $\sum_{i =1}^M u_{i}=0$. This necessitates that
\begin{equation}
    {x}_i - {u}_i = x_j - u_j \quad \forall i,j \in \{1,\dots,M\}
\end{equation}

To find the $\boldsymbol{u}$ that satisfies this condition, we can set
\begin{equation}
    {x}_i - {u}_i = c \quad \forall i \in \{1,\dots,M\}
\end{equation}
where $c$ is a constant signal. Then, by summing for $\forall i \in \{1,\dots,M\}$, we have

\begin{align}
    \sum _{i=1} ^M x_i - \sum _{i=1} ^M u_i &= Mc\\
    c &= \frac{1}{M} \sum _{i=1} ^M x_i\\
    x_i - u_i &= \frac{1}{M} \sum _{i=1} ^M x_i = \operatorname{j}^{-1} (\operatorname{P}_{\boldsymbol{C}}(\boldsymbol{x}))
\end{align}
where the operator $\operatorname{j}^{-1}$ is added to match the dimensionality.
\begin{equation}
    y_i = y = \operatorname{J}_{{\operatorname{A}}}(x_i - u_i) = \operatorname{J}_{{\operatorname{A}}}(\operatorname{j}^{-1} (\operatorname{P}_{\boldsymbol{C}}(\boldsymbol{x})))
\end{equation}
and thus
\begin{equation}
    \boldsymbol{y} = \operatorname{J}_{\boldsymbol{\operatorname{A}}}(\operatorname{P}_{\boldsymbol{C}}(\boldsymbol{x}))
\end{equation}
and this concludes the proof.
\end{proof}

\begin{theorem}{\label{thm:final}}  
    Let ${\mathcal{D}}$ be a closed, convex and non-empty subset of ${\mathcal{H}}$ with $\boldsymbol{\mathcal{D}} = \bigoplus_{i \in \{1,...,M\}} \mathcal{D}$ and, $\alpha > 0$ and let $\operatorname{E}$, $\operatorname{F}_i$ and $\operatorname{G}_i : \mathcal{H} \xrightarrow{} \mathcal{H}$, satisfy the assumptions:

    \begin{enumerate}
        \item Operator $\operatorname{E}$ is $\rho$-strongly monotone on $\mathcal{D}$
        \item Operator  $\operatorname{F}_i$ is $\gamma$-strongly monotone on $\mathcal{D}$
        \item Operator $\operatorname{G}_i$ satisfies
        \begin{equation}
             (\beta + \frac{1}{\epsilon})\norm{y - \hat{y}}^2 \geq \ip{\operatorname{G}_i y - \operatorname{G}_i \hat{y}}{y-\hat{y}} \geq \beta \norm{y - \hat{y}}^2
        \end{equation}
        for all $y, \hat y \in \mathcal{D}$.
        \item Operator $\boldsymbol{\operatorname{T}}$ defined in (\ref{eq:FPI}) satisfies $ \boldsymbol{\operatorname{T}}(\boldsymbol{\mathcal{D}}) \subseteq \boldsymbol{\mathcal{D}}$
        \item The condition $\gamma > \beta + 1 + \frac{2}{\epsilon}$ holds on the subset ${\mathcal{D}}$ for operator pairs $\operatorname{F}_i$ and $\operatorname{G}_i$
        \item $\text{Fix} \boldsymbol{\operatorname{T}} \neq \emptyset$
    \end{enumerate}
    then the shadow sequence of the fixed-point iteration $\boldsymbol{z}^{k+1} = \boldsymbol{\operatorname{T}}\boldsymbol{z}^k$ converges weakly to a  zero of (\ref{eq:method_app}).
\end{theorem}

\begin{proof}
 The proof begins with multiplying both sides of (\ref{eq:method_app}) by $M$ so

\begin{equation}{\label{eq:multiplied}}
    \left( M\operatorname{E} + M\sum_{i = 1}^M (\operatorname{F}_i - \operatorname{G}_i) \right)x=0
\end{equation}

Using the definitions of (\ref{eq:definitions}) and \cite[part (vii) of Proposition 26.4]{bauschke_convex_2011}, the representation of (\ref{eq:multiplied}) can be lifted to the product Hilbert space as
\begin{equation}
    \operatorname{j} \left( \operatorname{zer}(M \operatorname{E} +M \sum_{i = 1}^M (\operatorname{F}_i - \operatorname{G}_i)) \right) = \operatorname{zer} \left( \boldsymbol{\operatorname{E}} + \boldsymbol{\operatorname{F}} - \boldsymbol{\operatorname{G}} + \operatorname{N}_{\boldsymbol{C}} \right)
\end{equation}
where
\begin{equation}{\label{eq:big_operator}}
    \boldsymbol{\operatorname{E}} =  \begin{bmatrix}
\operatorname{E} \\
\operatorname{E} \\
\vdots  \\
\operatorname{E}
\end{bmatrix}, \quad \boldsymbol{\operatorname{F}} =  \begin{bmatrix}
M \operatorname{F}_1 \\
M \operatorname{F}_2 \\
\vdots  \\
M \operatorname{F}_M \\
\end{bmatrix}, \quad \boldsymbol{\operatorname{G}} =  \begin{bmatrix}
M \operatorname{G}_1 \\
M \operatorname{G}_2 \\
\vdots  \\
M \operatorname{G}_M \\
\end{bmatrix}
\end{equation}
are the new operators in the lifted product Hilbert space. 

It is convenient to define a new operator $\boldsymbol{\operatorname{S}} = \boldsymbol{\operatorname{E}} + \operatorname{N}_{\boldsymbol{C}}$. The problem at hand thus reduces to finding

\begin{equation}{\label{eq:equivalent}}
    \operatorname{zer}(\boldsymbol{\operatorname{S}} + \boldsymbol{\operatorname{F}} - \boldsymbol{\operatorname{G}})
\end{equation}

By using Theorem \ref{thm:three}, it is possible to derive the FPI algorithm that solves this problem. Maximal monotonicity of the operators and \cite[part (vi) of Proposition 26.4]{bauschke_convex_2011} permits the treatment of this problem as $M$ parallel ZFP in the product space of $\boldsymbol{\mathcal{H}}$. By using (\ref{eq:algorithm}), the first step is

\begin{equation}{\label{eq:first_step}}
    \begin{aligned}
    & \boldsymbol{x}^{k+1/2} = \operatorname{J}_{\alpha \boldsymbol{\operatorname{S}}}(\boldsymbol{z}^k) = \operatorname{J}_{\alpha (\boldsymbol{\operatorname{E}} + \operatorname{N}_{\boldsymbol{C}})}(\boldsymbol{z}^k)
    \end{aligned}
\end{equation}
where the special structure of $\boldsymbol{\operatorname{E}}$ as defined in (\ref{eq:big_operator}) and the properties of $\operatorname{N}_{\boldsymbol{C}}$ permit easy calculation of $\operatorname{J}_{\alpha (\boldsymbol{\operatorname{E}} + \operatorname{N}_{\boldsymbol{C}})}$.

Using the statement of Lemma \ref{lemma}, (\ref{eq:first_step}) reduces to 

\begin{equation}{\label{eq:rand}}
    \begin{aligned}
    & x_i^{k+1/2} =
    \operatorname{J}_{\alpha \operatorname{E}_i}(\operatorname{j}^{-1}(\operatorname{P}_{\boldsymbol{C}}(\boldsymbol{z}^k)))
    \end{aligned}
\end{equation}
and since by definition $\operatorname{E}_i = \operatorname{E}$, the subscript in (\ref{eq:rand}) can be dropped.

To proceed with the second step of (\ref{eq:algorithm}), the problem utilizes the parallel $M$ Hilbert spaces. Thus, the second step boils down to

\begin{equation}
    x^{k+1}_i=\operatorname{J}_{\alpha M \operatorname{F}_i}\left(2 x^{k+1 / 2}_i-z^{k}_i+\alpha M \operatorname{G}_i x^{k+1 / 2}_i\right)
\end{equation}
and the last step remain the same as
\begin{equation}
    z^{k+1} _i = z^{k} _i - x^{k+1/2} + x_i^{k+1}
\end{equation}

This proves the convergence of the FPI algorithm of Theorem \ref{thm:main}.
\end{proof}

\begin{remark}
    Both in theorem \ref{thm:three} and \ref{thm:final}, the property of strong monotonicity has been assumed for the first operator (e.g. $\operatorname{A}$ for the case of three operators and $\boldsymbol{\operatorname{E}}$ in the lifted case). Nevertheless, as shown in (\ref{eq:losslessness}), this requirement can be relaxed to losslessness without impacting the convergence of the problem.
\end{remark}

\bibliographystyle{IEEEtran}
\bibliography{ref}

\begin{IEEEbiography}[{\includegraphics[width=1in,height=1.25in,clip,keepaspectratio]{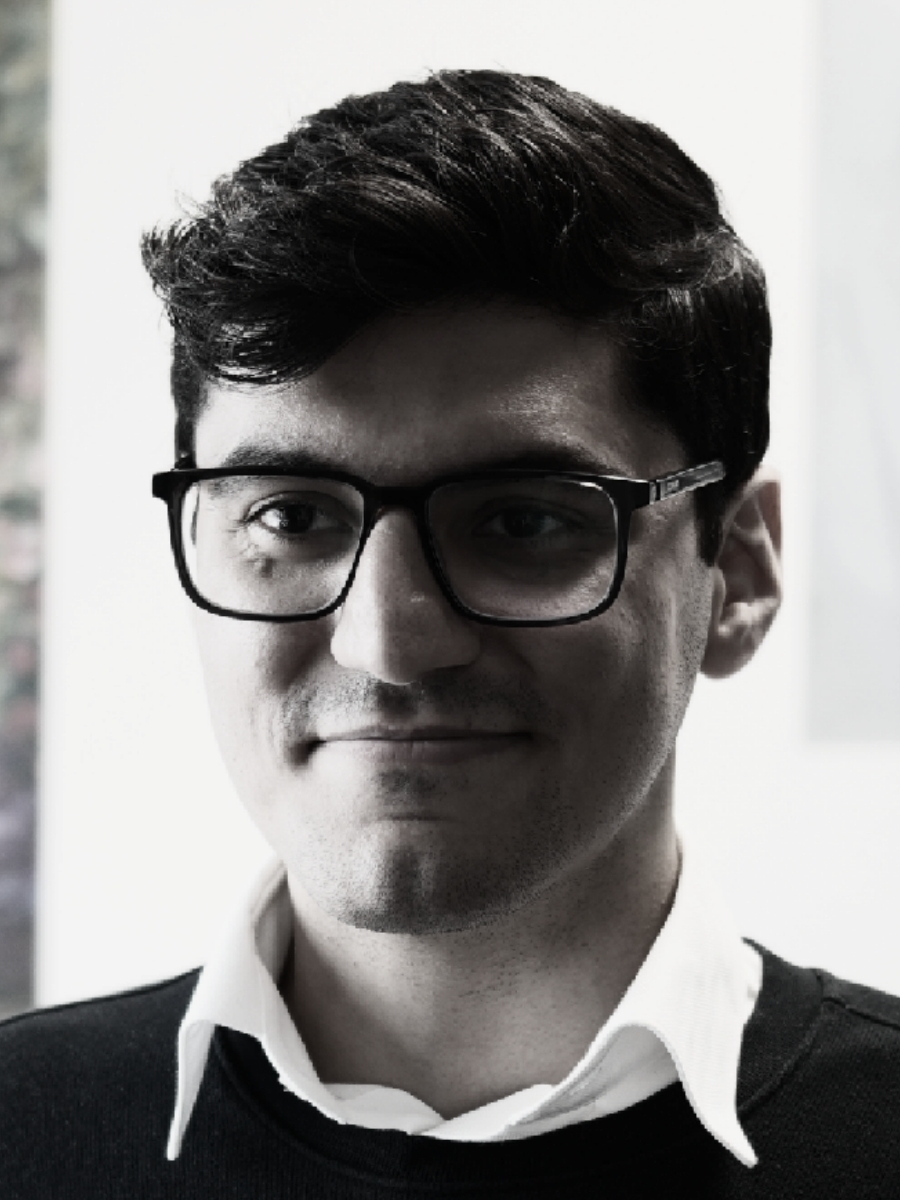}}]{Amir Shahhosseini} (Graduate Student Member, IEEE) received his M.Sc. degree in mechanical engineering from Sharif University of Technology in 2022. He is currently pursuing his Ph.D. at KU Leuven under the supervision of Prof. Rodolphe Sepulchre.  

    His research interests are system theory, optimization, and spiking neural networks. He is the recipient of the Harding distinguished scholarship, Cambridge, UK.
\end{IEEEbiography}

\begin{IEEEbiography}[{\includegraphics[width=1in,height=1.25in,clip,keepaspectratio]{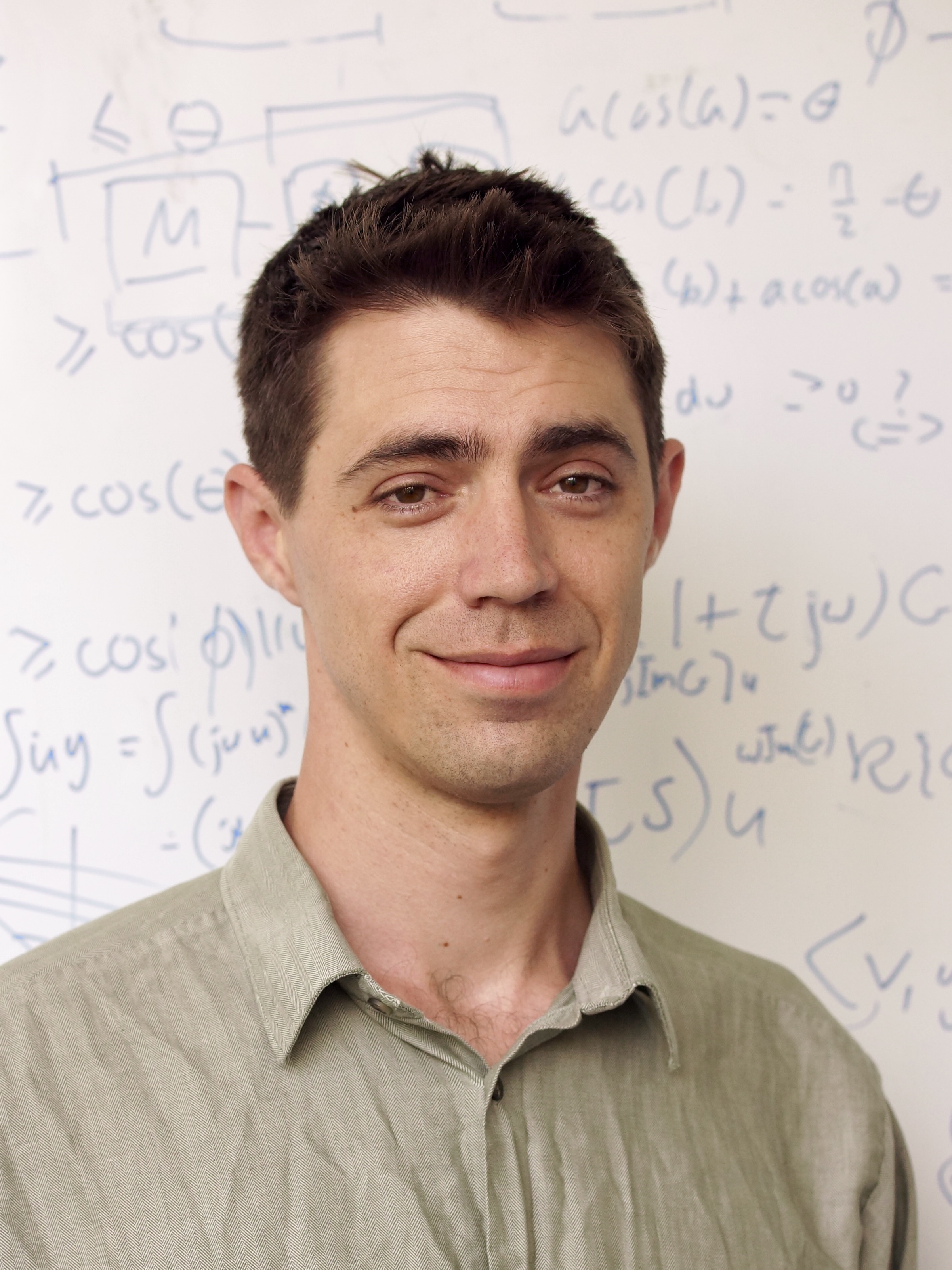}}]
{Thomas Chaffey} is a lecturer in the School of Electrical and Computer Engineering at the University of Sydney, Australia. He received the B.Sc. (advmath) degree in mathematics and computer science and the M.P.E. degree in mechanical engineering from the University of Sydney in 2015 and 2018, respectively, and the Ph.D. degree from the University of Cambridge, U.K., in 2022.  From 2022 to 2025 he held the Maudslay-Butler Research Fellowship in Engineering at Pembroke College, University of Cambridge. His research interests are in nonlinear control and its intersection with optimization, circuit theory and learning.
\end{IEEEbiography}

\begin{IEEEbiography}[{\includegraphics[width=1in,height=1.25in,clip,keepaspectratio]{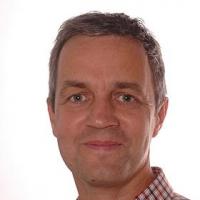}}]{Rodolphe Sepulchre} (Fellow, IEEE)  is professor of engineering at the KU Leuven (Belgium) and at the University of  Cambridge (UK). He is a fellow of IFAC (2020), IEEE (2009), and SIAM (2015).  He received the IEEE CSS Antonio Ruberti Young Researcher Prize in 2008 and the IEEE CSS George S. Axelby Outstanding Paper Award in 2020. He was elected at the Royal Academy of Belgium in 2013. He co-authored the monographs Constructive Nonlinear Control (1997, with M. Jankovic and P. Kokotovic) and Optimization on Matrix Manifolds (2008, with P.-A. Absil and R. Mahony).  His current research interests are focused in nonlinear control, with a focus on the feedback control principles of neuronal circuits. His research is currently funded by the ERC advanced grant SpikyControl (2023-2028).
\end{IEEEbiography}

\end{document}